\documentclass[envcountsect,runningheads,a4paper]{llncs}

\usepackage[titletoc]{appendix}
\usepackage{amssymb}
\usepackage{amsmath}
\usepackage{graphicx}
\usepackage{enumerate}
\usepackage{enumitem}
\usepackage{mathtools}
\usepackage{multicol}
\usepackage{proof}
\usepackage[normalem]{ulem}
\usepackage{cite}
\usepackage{tikz}
\usepackage{tikz-cd}
\usetikzlibrary{patterns,arrows, topaths, calc, positioning}
\tikzset{>=stealth}
\tikzstyle{node} = [circle, minimum size = 1.4mm, inner sep = 0mm, fill]
\tikzstyle{hyperedge} = [rectangle, minimum width = 5mm, minimum height = 5mm, draw, inner sep = 0mm]
\tikzstyle{HG} = [align = center]
\tikzstyle{circledge} = [circle, minimum size = 7mm, inner sep = 0mm, color=black, draw]

\newcommand{\eqdef}{\mathrel{\mathop:}=}
\newcommand{\SG}{\mathop{\mathrm{sg}}}
\newcommand{\LC}{\mathrm{L}}
\newcommand{\rk}{\mathit{rk}}

\newcommand{\HL}{\mathrm{HL}}

\newcommand{\DPO}{\mathrm{DPO}}

\begin{document}
	\title{Expressive Power of Hypergraph Lambek Grammars}
	\author{Tikhon Pshenitsyn \orcidID{0000-0003-4779-3143}\\ ptihon@yandex.ru}
	\authorrunning{T. Pshenitsyn}
	\institute{Steklov Mathematical Institute of Russian Academy of Sciences}
	\maketitle 

\begin{abstract}
	Hypergraph Lambek grammars (HL-grammars) is a novel logical approach to generating graph languages based on the hypergraph Lambek calculus. In this paper, we establish a precise relation between HL-grammars and hypergraph grammars based on the double pushout (DPO) approach: we prove that HL-grammars generate the same class of languages as DPO grammars with the linear restriction on lengths of derivations. This can be viewed as a complete description of the expressive power of HL-grammars and also as an analogue of the Pentus theorem, which states that Lambek grammars generate the same class of languages as context-free grammars. As a corollary, we prove that HL-grammars subsume contextual hyperedge replacement grammars.
\end{abstract}

\section*{Funding}

This work was supported by the Russian Science Foundation under grant no. 23-11-00104, https://rscf.ru/en/project/23-11-00104/.

\section{Introduction}\label{sec_introduction}
The Lambek calculus $\LC$ is a logic introduced by Joachim Lambek in \cite{Lambek58} to model the syntax of natural languages. The idea behind it is based on the concept of syntactic categories: certain groups of words (and of word collocations) behave similarly at the syntactic level and are interchangeable in contexts without affecting grammatical correctness; e.g. the words \textit{Arthur} and \textit{Gabrielle} along with the expression \textit{the girl from Ipanema} are of the same syntactic category (which we call \textit{noun phrase}). Lambek calculus formalizes syntactic categories as follows: one assigns a formula of $\LC$ to each word and also selects a distinguished formula $S$ to denote the syntactic category of sentences; then a sequence of words $a_1\dotsc a_n$ is accepted iff, after replacing each word $a_i$ by a formula (\textit{category}) $T_i$ assigned to it, the sequent $T_1,\dotsc,T_n \to S$ is derivable in $\LC$. This sequent can be informally understood as the following statement: ``a sequence of words of categories $T_1,\dotsc,T_n$ forms a string of category $S$, i.e. a sentence''. Formal grammars using such a mechanism for generating languages are called \textit{categorial grammars}.

The language of the Lambek calculus includes three binary connectives, which are divisions $\backslash$, $/$, and the product $\cdot$. Intuitively, a category of the form $A/B$ is assigned to a word if, whenever an expression of the type $B$ appears to the \textit{right} of it, they together form an expression of the category $A$; similarly for the category $B \backslash A$ (change \textit{right} to \textit{left}). For example, if $s$ is the category representing sentences, $n$ represents common nouns (e.g. \textit{girl}), and $np$ represents noun phrases denoting specific objects or subjects (e.g. \textit{Arthur}), then the article \textit{the} must have the category $np/n$ since it receives a common noun from the right and forms a noun phrase, and transitive verbs like \textit{loves} must have the category $(np\backslash s)/np$: they receive a noun phrase from the right and then another one from the left in order to form a sentence. In the Lambek calculus, one can derive the sequent $np,(np\backslash s)/np,np/n,n \to s$ as follows (see details in Section \ref{ssec_HL_HLG}):
\begin{equation}\label{eq_Altg}
	\infer[{\scriptstyle (/L) }]
	{
		np,(np\backslash s)/np, np/n,n \to s
	}
	{
		\infer[{\scriptstyle (\backslash L) }]{
			np,np\backslash s \to s
		}{
			s \to s
			&
			np \to np
		}
		&
		\infer[{\scriptstyle (/ L) }]{
			np/n,n \to np
		}{
			np \to np & n \to n
		}
	}
\end{equation}
This justifies the correctness of the sentence \textit{Arthur loves the girl}. 

In \cite{Pshenitsyn22_2}, we introduced the hypergraph Lambek calculus $\HL$, which generalizes the categorial paradigm to the field of graph grammars. The idea is the same: one assigns formulas of $\HL$ to labels of hyperedges of a hypergraph and then checks whether a corresponding \textit{hypergraph sequent} can be derived in $\HL$. This gives rise to hypergraph Lambek grammars ($\HL$-grammars). They are NP-complete and they subsume hyperedge replacement grammars \cite{Pshenitsyn22_2}, which are considered to be a hypergraph extension of context-free grammars. 

In \cite{Pentus93}, Mati Pentus proved that grammars over $\LC$ generate the same languages as context-free grammars. The proof of this result uses quite delicate techniques, which are hard to adapt for variants of the Lambek calculus. In the hypergraph case, moreover, $\HL$-grammars and hyperedge replacement grammars are not equivalent: the former are strictly more expressive than the latter. In \cite{Pshenitsyn22_2}, we proved that $\HL$-grammars generate finite intersections of languages generated by hyperedge replacement grammars. 

Continuing the attempts to find a bridge between $\HL$-grammars and existing graph grammars, we looked at the hypergraph grammars based on the double-pushout approach (\emph{DPO grammars}) \cite{EhrigPS73}. In \cite{Pshenitsyn22_1}, we proved that $\HL$-grammars are at least as powerful as DPO grammars \textit{with linear restriction on lengths of derivations}, i.e. such that one takes into account only derivations with the number of steps bounded by a linear function w.r.t. the number of hyperedges of the resulting hypergraph (see more in Section \ref{sec_linDPO}). This was done using methods similar to those used by Kanazawa to transform an unrestricted Chomsky grammar into a grammar over the multiplicative-exponential Lambek calculus \cite{Kanazawa99}.

In this paper, we establish that hypergraph Lambek grammars are equivalent to the DPO grammars with the linear restriction (Section \ref{sec_linDPO}). Since the latter represent a straightforward and natural modification of DPO grammars, this result may be viewed as an analogue of the Pentus theorem. It should be noticed that linearly restricted DPO grammars have not been studied in the existing literature (although a similar notion is defined for string grammars in \cite{Book71}), so little is known about their possibilities and limitations. This result has many nice corollaries; due to page limits we will discuss only the following one in Section \ref{ssec_rem_and_cor}: $\HL$-grammars generate all languages generated by contextual hyperedge replacement grammars. In Section \ref{sec_conclusion}, we give concluding remarks.

\section{Preliminaries}\label{sec_preliminaries}
$\mathcal{T}^\ast$ is the set of strings over an alphabet $\mathcal{T}$ including the empty word $\Lambda$; if $R$ is a relation, then $R^\ast$ is its transitive reflexive closure. $A\sqcup B$ denotes the disjoint union of the sets $A$ and $B$. Each function $f:\mathcal{T} \to \Delta$ can be extended to a homomorphism $f:\mathcal{T}^\ast \to \Delta^\ast$. By $w(i)$ we denote the $i$-th symbol of $w \in \mathcal{T}^\ast$; $|w|$ is the number of symbols in $w$; $[n]$ denotes the set $\{1,2,\dotsc,n\}$ (and $[0] \eqdef \emptyset$).

A \emph{ranked set} $M$ is the set along with a \emph{rank function} $\rk:M \to \mathbb{N}$. Given a ranked set of \emph{labels} $\mathcal{T}$, a \emph{hypergraph} $G$ over $\mathcal{T}$ is a tuple $\langle V_G, E_G, att_G, lab_G, ext_G \rangle$ where $V_G$ is a finite set of \emph{nodes}, $E_G$ is a finite set of \emph{hyperedges}, $att_G: E_G\to V_G^\ast$ assigns a string (i.e. an ordered multiset) of \emph{attachment nodes} to each hyperedge, $lab_G: E_G \to \mathcal{T}$ labels each hyperedge by some element of $\mathcal{T}$ in such a way that $\rk(lab_G(e))=|att_G(e)|$ whenever $e\in E_G$, and $ext_G\in V_G^\ast$ is a string of \emph{external nodes}.  Hypergraphs are always considered up to isomorphism. The set of all hypergraphs with labels from $\mathcal{T}$ is denoted by $\mathcal{H}(\mathcal{T})$. The \emph{rank function} $\rk_G$ (or $\rk$, if $G$ is clear) is defined as follows: $\rk_G(e)\eqdef|att_G(e)|$. Besides, $\rk(G)\eqdef |ext_G|$.

In drawings of hypergraphs, small circles are nodes, labeled rectangles are hyperedges, $att_G$ is represented by numbered lines, and external nodes are represented by numbers in parentheses (round brackets). If a hyperedge has exactly two attachment nodes, then it is depicted by a labeled arrow that goes from the first attachment node to the second one.

A \emph{handle} $a^\bullet$ is a hypergraph $\langle [n],[1],att,lab,1\dotsc n\rangle$ where $att(1)=1\dotsc n$ and $lab(1)=a$ ($a\in\mathcal{T}$, $rk(a)=n$). Hypergraphs without hyperedges are called \emph{edgeless}. \emph{A string graph $\SG(w)$ induced by a string $w=a_1\dots a_n$} is the hypergraph $\langle \{v_i\}_{i=0}^n,\{e_i\}_{i=1}^n,att,lab,v_0v_n \rangle$ where $att(e_i)=v_{i-1}v_i$, $lab(e_i)=a_i$.

Given a hypergraph $H$ and a function $f:E_H\to \mathcal{T}$, \emph{a relabeling $f(H)$} is the hypergraph $f(H)=\langle V_H, E_H, att_H, f, ext_H\rangle$. It is required that $rk_H(e)=rk(f(e))$ for $e\in E_H$. 

The \emph{replacement of a hyperedge $e_0$ in $G$ ($e_0 \in E_G$) by a hypergraph $H$} (such that $rk(e_0)=rk(H)$) is done as follows: (1) remove $e_0$ from $G$; (2) insert an isomorphic copy of $H$ ($H$ and $G$ must consist of disjoint sets of nodes and hyperedges); (3) for each $i=1,\dotsc,rk(e_0)$, fuse the $i$-th external node of $H$ with the $i$-th attachment node of $e_0$ (formally, the set of new nodes is $(V_G \sqcup V_H)/\equiv$ where $\equiv$ is the smallest equivalence relation satisfying $att_G(e_0)(i)\equiv ext_H(i)$). The result is denoted as $G[e_0/H]$. If several hyperedges of a hypergraph are replaced by other hypergraphs, then the result does not depend on the order of the replacements; moreover the result does not change, if replacements are done simultaneously \cite{DrewesKH97}. The following notation is in use: if $e_1,\dots,e_k$ are distinct hyperedges of a hypergraph $H$ and they are simultaneously replaced by hypergraphs $H_1,\dots,H_k$ resp., then the result is denoted $H[e_1/H_1,\dots,e_k/H_k]$.

The \emph{disjoint union} $H_1+H_2$ of the hypergraphs $H_1$, $H_2$ as the hypergraph $\langle V_{H_1}\sqcup V_{H_2}, E_{H_1}\sqcup E_{H_2}, att, lab, ext_{H_1}ext_{H_2}\rangle$ such that $att|_{H_i} = att_{H_i}$, $lab|_{H_i} = lab_{H_i}$ ($i=1,2$); that is, we just put these hypergraphs together without fusing any nodes or hyperedges. The disjoint union of a hypergraph $H$ (such that $\rk(H)=0$) with itself $k$ times is denoted by $k \cdot H$. The multiple disjoint union $H_1 + \dotsc + H_k$ can be shortly denoted as $\sum_{i=1}^k H_i$. Note that in general $H+G \ne G+H$ but $H+G=G+H$ holds whenever $\rk(H)=0$.

\subsection{DPO Hypergraph Grammars}
The general definitions of the double pushout rewriting can be found in various textbooks and articles (see e.g. \cite{KonigNPR18}). To recall, a rule in the double pushout approach is of the form $L \overset{\varphi_L}{\leftarrow} I \overset{\varphi_R}{\rightarrow} R$ where $L,I,R$ are hypergraphs of rank 0 and $\varphi_L$, $\varphi_R$ are morphisms; then an application of this rule to a hypergraph $G$ consists of finding a subhypergraph matching $L$ within $G$ and then replacing $L$ by $R$. In contrast with hyperedge replacement, it is not clear how to replace a hypergraph by another one; this is defined using the categorical notion of double pushout. The interface hypergraph $I$ plays the role of ``the shared part'' of $L$ and $R$, and it specifies how exactly $L$ is replaced by $R$. 

In this paper, we are not interested in working on the level of generality category theory provides us, so we consider DPO grammars only within the category of hypergraphs. Moreover, we consider only rules where $I$ is discrete. This leads us to the following simple definition of a hypergraph DPO rule:

\begin{definition}
	A \emph{DPO rule over a set of labels $C$} is of the form $r = (L | R)$ where $L,R \in \mathcal{H}(C)$ are hypergraphs such that $\rk(L)=\rk(R)$. 
	\\
	A hypergraph $G$ \emph{is transformed into $H$ via $r=(L|R)$} if there is a hypergraph $C$ with a distinguished hyperedge $e_0$ such that $G=C[e_0/L]$ and $H=C[e_0/R]$. This is denoted as follows: $G \Rightarrow_r H$ or $G \Rightarrow H$. If $H$ is obtained from $G$ by applying $k$ rules from some fixed set, then we write $G \Rightarrow^k H$.
\end{definition}
\begin{example}\label{ex_dpo_rule}
	Consider the DPO rule $\rho$ and an example of its application:
	\begin{equation}\label{ex_application_dpo_rule}
		\rho = \left( \!
		{\color{orange}
			\vcenter{\hbox{{\tikz[baseline=.1ex]{
							\node[node, label=above:{\scriptsize $(1)$}] (N) {};
							\node[node, below left=4mm and 4mm of N, label=left:{\scriptsize $(2)$}] (N1) {};
							\node[node, below right=4mm and 4mm of N, label=right:{\scriptsize $(3)$}] (N2) {};
							\draw[->] (N) -- node[above left] {$l$} (N1);
							\draw[->] (N) -- node[above right] {$r$} (N2);
			}}}}
		}
		\middle|
		{\color{teal}
			\vcenter{\hbox{{\tikz[baseline=.1ex]{
							\node[node] (N) {};
							\node[node, below = 4mm of N, label=left:{\scriptsize $(2)$}] (N1) {};
							\node[node, right=11mm of N1, label=right:{\scriptsize $(3)$}] (N2) {};
							\node[node, right=11mm of N, label=right:{\scriptsize $(1)$}] (N3) {};
							\node[hyperedge, right = 2.5mm of N] (F) {$f$};
							\draw[->] (N1) -- node[below] {$t$} (N2);
							\draw[-] (N) -- node[above] {\scriptsize 1} (F);
			}}}}
		} \right)
		\quad
		\boxed{
		\!\!
		\vcenter{\hbox{{\tikz[baseline=.1ex]{
						\node[orange,node] (N) {};
						\node[orange,node, below left = 4mm and 4mm of N, label=left:{\scriptsize $(1)$}] (N1) {};
						\node[orange,node, below right = 4mm and 4mm of N] (N2) {};
						\draw[->] (N2) -- node[below] {$q$} (N1);
						\draw[orange,->] (N) to[bend right = 0] node[above] {$l$} (N1);
						\draw[orange,->] (N) to[bend left = 0] node[above] {$r$} (N2);
		}}}}
		\;\,\Rightarrow\;\,
		\vcenter{\hbox{{\tikz[baseline=.1ex]{
						\node[teal,node] (N) {};
						\node[teal,node, below right=4.7mm and 4mm of N] (N2) {};
						\node[teal,node, below left=4.7mm and 4mm of N, label=below left:{\color{black}\scriptsize $(1)$}] (N1) {};
						\node[teal,hyperedge,left=2.5mm of N] (E2) {$f$};
						\node[teal,node, left=2mm of E2] (N3) {};
						\draw[teal,-] (N3) -- node[below] {\scriptsize 1} (E2);
						\draw[teal,->] (N1) to[bend right = 50] node[below] {$t$} (N2);
						\draw[->] (N2) to[bend right = 0] node[above right] {$q$} (N1);
		}}}}
		}
	\end{equation}
\end{example}
\begin{definition}
	A \emph{DPO hypergraph grammar} $HGr$ is of the form $\langle \mathcal{N},\mathcal{T},\mathcal{P}, Z\rangle$ where $N$, $\mathcal{T}$ are disjoint finite alphabets of \emph{nonterminal} and \emph{terminal} labels resp., $\mathcal{P}$ is a finite set of hypergraph grammar rules over $\mathcal{N} \cup \mathcal{T}$ and $Z$ is the \emph{start hypergraph}. The language $L(HGr)$ generated by $HGr$ is the set of all hypergraphs $H \in \mathcal{H}(\mathcal{T})$ such that $Z \Rightarrow^\ast H$ using the rules from $\mathcal{P}$.
\end{definition}

\subsection{Hypergraph Lambek Calculus and Its Grammars}\label{ssec_HL_HLG}
Now let us define the hypergraph Lambek calculus $\HL$ and $\HL$-grammars (see discussion of the definitions in \cite{Pshenitsyn21_2, Pshenitsyn22_2}). Firstly we recall how the standard Lambek calculus is organized and then we proceed with formal definitions of $\HL$.

The Lambek calculus $\LC$ \cite{Lambek58, MootR12} is a propositional substructural logic. Its language includes two divisions $\backslash$ and $/$, which can be viewed as directed implications, and the product $\cdot$. The Lambek calculus works with sequents, which are structures of the form $A_1,\dotsc,A_n \to B$ where $A_i$, $B$ are categories. Usually, it is required that $n>0$, which is called Lambek's restriction (it arises for linguistic purposes). If $n\ge 0$ is allowed, then the resulting calculus is \emph{the Lambek calculus allowing empty premises} $\LC^\ast$ \cite{Pentus98}. The only axiom of $\LC$ (and of $\LC^\ast$) is $A \to A$; the rules are the following (for $\LC$, we require that $\Pi$, $\Psi$ are non-empty):
$$
\infer[{\scriptstyle ( \backslash L)}]{\Gamma, \Pi, A \backslash B, \Delta \to C}{\Gamma, B, \Delta \to C & \Pi \to A}
\quad
\infer[{\scriptstyle ( \backslash R)}]{\Pi \to A \backslash B}{A, \Pi \to B}
\quad
\infer[{\scriptstyle ( \cdot L)}]{\Gamma, A \cdot B, \Delta \to C}{\Gamma, A, B, \Delta \to C}
$$
$$
\infer[{\scriptstyle ( / L)}]{\Gamma, B / A, \Pi, \Delta \to C}{\Gamma, B, \Delta \to C & \Pi \to A}
\quad
\infer[{\scriptstyle ( / R)}]{\Pi \to B / A}{\Pi, A \to B}
\quad
\infer[{\scriptstyle ( \cdot R)}]{\Pi, \Psi \to A \cdot B}{\Pi \to A & \Psi \to B}
$$
\begin{example}\label{ex_der_L_2}
	The sequent $p/q \to (p\cdot r)/(q\cdot r)$ is derivable in $\LC$:
	\begin{equation}
		\infer[{\scriptstyle (/ R) }]
		{
			p/q \to (p\cdot r)/(q\cdot r)
		}
		{
			\infer[{\scriptstyle (\cdot L) }]{
				p/q, q\cdot r \to p\cdot r
			}{
				\infer[{\scriptstyle (/ L) }]{
					p/q, q, r \to p\cdot r
				}{
					\infer[{\scriptstyle (\cdot R) }]{
						p, r \to p\cdot r
					}{
						p \to p
						&
						r \to r
					}
					&
					q \to q
				}
			}
		}
	\end{equation}
\end{example}
A Lambek grammar consists of a finite alphabet $\mathcal{T}$, a finite binary relation $\triangleright$ between symbols of the alphabet and formulas (categories) of $\LC$, and of a distinguished formula $S$. A string $a_1\dotsc a_n \in \mathcal{T}^\ast$ belongs to the language of this grammar iff for each $i$ there exists a category $T_i$ such that $a_i \triangleright T_i$ and such that the sequent $T_1,\dotsc,T_n \to S$ is derivable in $\LC$. Returning to the example (\ref{eq_Altg}), there $\mathcal{T} = \{\mbox{\textit{Arthur}}, \mbox{\textit{loves}}, \mbox{\textit{the}}, \mbox{\textit{girl}}\}$, $\mbox{\textit{Arthur}} \triangleright np$, $\mbox{\textit{loves}} \triangleright (np \backslash s)/np$, $\mbox{\textit{the}} \triangleright np/n$, $\mbox{\textit{girl}} \triangleright n$, and $S=s$ is a primitive category.

In our works \cite{Pshenitsyn21_2,Pshenitsyn22_2} we generalized the Lambek calculus to hypergraphs trying to develop a logic-based graph grammar formalism; the motivation was mostly theoretical, i.e. we were interested in mathematical properties of the resulting formalism compared to other graph grammars and to the string Lambek calculus. The resulting generalization was designed by comparing context-free grammars, Lambek grammars and hyperedge replacement grammars. The hypergraph Lambek calculus has categories, sequents and rules that work in a similar way to those of $\LC$. In fact, $\LC$ can be embedded in $\HL$ if one interprets strings as string graphs. 

To define the set of formulas (which are called \textit{types}) of $\HL$, we fix a ranked set $\mathit{Pr}$ of \emph{primitive types} such that for each $k \in\mathbb{N}$ there are infinitely many $p \in \mathit{Pr}$ satisfying $\rk(p) = k$. Besides, we fix a countable set of labels $\$_n,n\in\mathbb{N}$ and set $\rk(\$_n)=n$; let us agree that these labels do not belong to any other set considered in the definition of the calculus. Then the set of \emph{types} $\mathit{Tp}$ is defined inductively as follows:
\begin{enumerate}
	\item All primitive types are types.
	\item Let $N \in \mathit{Tp}$ be a type, and let $D$ be a hypergraph such that labels of all its hyperedges, except for one, are from $\mathit{Tp}$, and the remaining one equals $\$_d$ for some $d$; let also $\rk(N)=\rk(D)$ \underline{and $|E_D|>1$}. Then $N\div D$ is also a type and $\rk(N\div D)\eqdef d$. The hyperedge of $D$ labeled by $\$_d$ is denoted by $e^\$_D$.
	\item If $M$ is a hypergraph labeled by types from $\mathit{Tp}$ \underline{such that $|E_M|>0$}, then $\times(M)$ is also a type, and $\rk(\times(M))\eqdef \rk(M)$.
\end{enumerate}
\begin{definition}
	A \emph{sequent} is a structure of the form $H\to A$ where $H \in \mathcal{H}(\mathit{Tp})$ is the \emph{antecedent} of the sequent \underline{such that $|E_H|>0$}, and $A \in \mathit{Tp}$ is the \emph{succedent} of the sequent such that $\rk(H)=\rk(A)$.
\end{definition} 

The hypergraph Lambek calculus $\HL$ derives sequents in the latter sense. The only axiom of $\HL$ is of the form $A^\bullet\to A$ where $A\in \mathit{Tp}$ ($A^\bullet$ is a handle). There are four inference rules of $\HL$:
		$$
		\infer[{\scriptstyle (\div L)}]{H\left[e/D[e^\$_D/ (N\div D)^\bullet,d_1/H_1,\dotsc,d_k/H_k]\right]\to A}{H[e/N^\bullet]\to A & H_1\to lab_D(d_1) &\dots & H_k\to lab_D(d_k)}
		$$
		$$
		\infer[{\scriptstyle(\times R)}]{M[m_1/H_1,\dotsc,m_l/H_l]\to\times(M)}{H_1\to lab_M(m_1) & \dots & H_l\to lab_M(m_l)}
		$$
		$$
		\infer[{\scriptstyle(\div R)}]{F\to N\div D}{D[e^\$_D/F]\to N}
		\qquad\qquad
		\infer[{\scriptstyle(\times L)}]{H[e/(\times(M))^\bullet]\to A}{H[e/M]\to A}
		$$
Here $N\div D$, $\times(M) \in \mathit{Tp}$; $e \in E_H$; $E_D=\{e^\$_D,d_1,\dots,d_k\}$, $E_M = \{m_1,\dotsc, m_l\}$. \underline{All antecedents of sequents must have at least one hyperedge.} 

A sequent $H\to A$ is said to be \emph{derivable in $\mathrm{HL}$} (denoted by $\mathrm{HL}\vdash H\to A$) if it can be obtained from axioms by applications of rules of $\mathrm{HL}$. 

\begin{definition}
	An \emph{$\HL$-grammar} is a tuple $HGr=\langle \mathcal{T}, S, \triangleright\rangle$ where $\mathcal{T}$ is a ranked alphabet, $S\in \mathit{Tp}$ is a distinguished type, and $\triangleright\subseteq\mathcal{T}\times \mathit{Tp}$ is a finite binary relation such that $a\triangleright T$ implies $\rk(a)=\rk(T)$. \emph{The language $L(HGr)$ generated by an $\HL$-grammar} $HGr=\langle \mathcal{T}, S, \triangleright\rangle$ contains a hypergraph $G\in\mathcal{H}(\mathcal{T})$ if and only if \underline{$|E_G|>0$} and a function $f_G:E_G\to \mathit{Tp}$ exists such that: 
	\begin{enumerate}
		\item $lab_G(e)\triangleright f_G(e)$ whenever $e\in E_G$;
		\item $\mathrm{HL}\vdash f_G(G)\to S$ (recall that $f_G(G)$ is a relabeling of $G$ by means of $f_G$).
	\end{enumerate}
\end{definition}

\begin{remark}\label{rem_empty_antecedents}
	The underlined parts of the above definitions represent Lambek's restriction in the hypergraph case: we forbid edgeless hypergraphs and we also require that for $N \div D \in \mathit{Tp}$ the denominator $D$ must have more hyperedges than just $e^\$_D$. In \cite{Pshenitsyn21_2,Pshenitsyn22_2}, there is another restriction instead of this: coincidences of external nodes of a hypergraph are forbidden. However, it seems to be less natural and less related to the original restriction, which requires that any antecedent contains at least one type. Moreover, allowing repeated external nodes is necessary to prove the main result of this work (Theorem \ref{th_main}, Example \ref{ex_translation_HL_to_DPO}). Thus, from now on, the official definition of $\HL$ is as in this paper. We can also define \emph{the hypergraph Lambek calculus with edgeless premises} $\HL^\ast$ (along with $\HL^\ast$-grammars) by removing all the underlined parts of the above definition.	
\end{remark}

\begin{example}
	Consider the following types ($s,n,np \in \mathit{Pr}$ are types of rank $2$):
	\begin{enumerate}
		\item $U^\prime_1=\tau(np\backslash s) = s \div 
		\left(\vcenter{\hbox{{\tikz[baseline=.1ex]{
						\node[node,label=left:{\scriptsize $(1)$}] (N1) {};
						\node[node,right=7mm of N1] (N2) {};
						\node[node, right=7mm of N2,label=right:{\scriptsize $(2)$}] (N3) {};
						\draw[->] (N1) -- node[above] {$np$} (N2);
						\draw[->] (N2) -- node[above] {$\$_2$} (N3);
		}}}}\right) = s \div \SG(\mathit{np} \, \$_2)$
		\item $U_1=\tau((np\backslash s)/np) = U_1^\prime \div 
		\left(\vcenter{\hbox{{\tikz[baseline=.1ex]{
						\node[node,label=left:{\scriptsize $(1)$}] (N1) {};
						\node[node,right=6mm of N1] (N2) {};
						\node[node, right=6mm of N2,label=right:{\scriptsize $(2)$}] (N3) {};
						\draw[->] (N1) -- node[above] {$\$_2$} (N2);
						\draw[->] (N2) -- node[above] {$np$} (N3);
		}}}}\right)$
		\item $U_2=\tau(np/n) = np \div 
		\SG(\$_2 \, n)$
	\end{enumerate}
	The sequent $\vcenter{\hbox{{\tikz[baseline=.1ex]{
					\node[node,label=left:{\scriptsize $(1)$}] (N1) {};
					\node[node,right=7mm of N1] (N2) {};
					\node[node,right=7mm of N2] (N3) {};
					\node[node,right=7mm of N3] (N4) {};
					\node[node, right=7mm of N4,label=right:{\scriptsize $(2)$}] (N5) {};
					\draw[->] (N1) -- node[above] {$np$} (N2);
					\draw[->] (N2) -- node[above] {$U_1$} (N3);
					\draw[->] (N3) -- node[above] {$U_2$} (N4);
					\draw[->] (N4) -- node[above] {$n$} (N5);
	}}}}
	\;\to\;
	s $ is derivable in $\HL$:
	\begin{equation}
		\infer[{\scriptstyle (\div L) }]{
			\vcenter{\hbox{{\tikz[baseline=.1ex]{
							\node[node,label=left:{\scriptsize $(1)$}] (N1) {};
							\node[node,right=6mm of N1] (N2) {};
							\node[node,right=6mm of N2] (N3) {};
							\node[node,right=6mm of N3] (N4) {};
							\node[node, right=6mm of N4,label=right:{\scriptsize $(2)$}] (N5) {};
							\draw[->] (N1) -- node[above] {$np$} (N2);
							\draw[->] (N2) -- node[above] {$U_1$} (N3);
							\draw[->] (N3) -- node[above] {$U_2$} (N4);
							\draw[->] (N4) -- node[above] {$n$} (N5);
			}}}}
			\to
			s 
		}{
			\infer[{\scriptstyle (\div L) }]{
				\vcenter{\hbox{{\tikz[baseline=.1ex]{
								\node[node,label=left:{\scriptsize $(1)$}] (N1) {};
								\node[node,right=6mm of N1] (N2) {};
								\node[node, right=6mm of N2,label=right:{\scriptsize $(2)$}] (N3) {};
								\draw[->] (N1) -- node[above] {$np$} (N2);
								\draw[->] (N2) -- node[above] {$U^\prime_1$} (N3);
				}}}}
				\to
				s 
			}{
				\vcenter{\hbox{{\tikz[baseline=.1ex]{
								\node[node,label=below:{\scriptsize $(1)$}] (N1) {};
								\node[node, right=6mm of N1,label=below:{\scriptsize $(2)$}] (N2) {};
								\draw[->] (N1) -- node[above] {$s$} (N2);
				}}}}
				\to
				s
				&
				\vcenter{\hbox{{\tikz[baseline=.1ex]{
								\node[node,label=below:{\scriptsize $(1)$}] (N1) {};
								\node[node, right=6mm of N1,label=below:{\scriptsize $(2)$}] (N2) {};
								\draw[->] (N1) -- node[above] {$np$} (N2);
				}}}}
				\to
				np
			}
			&
			\infer[{\scriptstyle (\div L) }]{
				\vcenter{\hbox{{\tikz[baseline=.1ex]{
								\node[node,label=left:{\scriptsize $(1)$}] (N1) {};
								\node[node,right=6mm of N1] (N2) {};
								\node[node, right=6mm of N2,label=right:{\scriptsize $(2)$}] (N3) {};
								\draw[->] (N1) -- node[above] {$U_2$} (N2);
								\draw[->] (N2) -- node[above] {$n$} (N3);
				}}}}
				\to
				np
			}{
				\vcenter{\hbox{{\tikz[baseline=.1ex]{
								\node[node,label=below:{\scriptsize $(1)$}] (N1) {};
								\node[node, right=6mm of N1,label=below:{\scriptsize $(2)$}] (N2) {};
								\draw[->] (N1) -- node[above] {$np$} (N2);
				}}}}
				\to
				np
				&
				\vcenter{\hbox{{\tikz[baseline=.1ex]{
								\node[node,label=below:{\scriptsize $(1)$}] (N1) {};
								\node[node, right=6mm of N1,label=below:{\scriptsize $(2)$}] (N2) {};
								\draw[->] (N1) -- node[above] {$n$} (N2);
				}}}}
				\to
				n
			}
		}	
	\end{equation}
	E.g. the last application of the rule ${\scriptstyle(\div L)}$ is done as follows: $N=U^\prime_1$; $H[e/N^\bullet]=\\\vcenter{\hbox{{\tikz[baseline=.1ex]{
					\node[node,label=left:{\scriptsize $(1)$}] (N1) {};
					\node[node,right=5mm of N1] (N2) {};
					\node[node, right=5mm of N2,label=right:{\scriptsize $(2)$}] (N3) {};
					\draw[->] (N1) -- node[above] {$np$} (N2);
					\draw[->] (N2) -- node[above] {$U^\prime_1$} (N3);
	}}}}$; $D = \vcenter{\hbox{{\tikz[baseline=.1ex]{
	\node[node,label=left:{\scriptsize $(1)$}] (N1) {};
	\node[node,right=5mm of N1] (N2) {};
	\node[node, right=5mm of N2,label=right:{\scriptsize $(2)$}] (N3) {};
	\draw[->] (N1) -- node[above] {$\$_2$} (N2);
	\draw[->] (N2) -- node[above] {$np$} (N3);
	}}}}$; $k=1$; $H_1 = \vcenter{\hbox{{\tikz[baseline=.1ex]{
	\node[node,label=left:{\scriptsize $(1)$}] (N1) {};
	\node[node,right=5mm of N1] (N2) {};
	\node[node, right=5mm of N2,label=right:{\scriptsize $(2)$}] (N3) {};
	\draw[->] (N1) -- node[above] {$U_2$} (N2);
	\draw[->] (N2) -- node[above] {$n$} (N3);
	}}}}$; $lab_D(d_1) = np$. Then we modify the antecedent as follows:
	\begin{center}
		\begin{tabular}{l|ccl|c}
			$H[e/N^\bullet]$
			&
			$\vcenter{\hbox{{\tikz[baseline=.1ex]{
							\node[node,label=below:{\scriptsize $(1)$}] (N1) {};
							\node[node,right=5mm of N1] (N2) {};
							\node[node, right=5mm of N2,label=below:{\scriptsize $(2)$}] (N3) {};
							\draw[->] (N1) -- node[above] {$np$} (N2);
							\draw[->] (N2) -- node[above] {$U^\prime_1$} (N3);
			}}}}$
			&
			&
			$H[e/D[e^\$_D/(N\!\div\! D)^\bullet]]$
			&
			$\vcenter{\hbox{{\tikz[baseline=.1ex]{
							\node[node,label=below:{\scriptsize $(1)$}] (N1) {};
							\node[node,right=5mm of N1] (N2) {};
							\node[node,right=5mm of N2] (N3) {};
							\node[node, right=5mm of N3,label=below:{\scriptsize $(2)$}] (N4) {};
							\draw[->] (N1) -- node[above] {$np$} (N2);
							\draw[->] (N2) -- node[above] {$U_1$} (N3);
							\draw[->] (N3) -- node[above] {$np$} (N4);
			}}}}$
			\\
			$H[e/D]$
			&
			$\vcenter{\hbox{{\tikz[baseline=.1ex]{
							\node[node,label=below:{\scriptsize $(1)$}] (N1) {};
							\node[node,right=5mm of N1] (N2) {};
							\node[node,right=5mm of N2] (N3) {};
							\node[node, right=5mm of N3,label=below:{\scriptsize $(2)$}] (N4) {};
							\draw[->] (N1) -- node[above] {$np$} (N2);
							\draw[->] (N2) -- node[above] {$\$_2$} (N3);
							\draw[->] (N3) -- node[above] {$np$} (N4);
			}}}}$
			&
			&
			$H[e/D[e^\$_D/(N \!\div\! D)^\bullet,d_1/H_1]]$
			&
			$\vcenter{\hbox{{\tikz[baseline=.1ex]{
							\node[node,label=below:{\scriptsize $(1)$}] (N1) {};
							\node[node,right=5mm of N1] (N2) {};
							\node[node,right=5mm of N2] (N3) {};
							\node[node,right=5mm of N3] (N4) {};
							\node[node, right=5mm of N4,label=below:{\scriptsize $(2)$}] (N5) {};
							\draw[->] (N1) -- node[above] {$np$} (N2);
							\draw[->] (N2) -- node[above] {$U_1$} (N3);
							\draw[->] (N3) -- node[above] {$U_2$} (N4);
							\draw[->] (N4) -- node[above] {$n$} (N5);
			}}}}$
			\\
		\end{tabular}
	\end{center}
\end{example}
In general, the function $\tau$ defined below is an embedding of $\LC$ in $\HL$:
\begin{multicols}{2}
	\begin{enumerate}
		\item $\tau(p)\eqdef p$, $p\in \mathit{Pr}$, $\rk(p)=2$;
		\item $\tau(A/B)\eqdef \tau(A) \div \SG(\$_2\;\tau(B))$;
		\item $\tau(B\backslash A)\eqdef \tau(A) \div \SG(\tau(B) \; \$_2)$;
		\item $\tau(A\cdot B)\eqdef \times \SG(\tau(A)\;\tau(B))$.
	\end{enumerate}
\end{multicols}
Namely, the following theorem is proved in \cite{Pshenitsyn22_2}: 
\begin{theorem}
	A sequent $A_1,\dotsc,A_n \to B$ is derivable in $\LC$ if and only if its translation $\SG(\tau(A_1)\dotsc \tau(A_n)) \to \tau(B)$ is derivable in $\HL$. 
\end{theorem}

\section{Linearly Restricted DPO Grammars}\label{sec_linDPO}
DPO grammars are more powerful than $\HL$-grammars since they are universal; the former are undecidable and $\Sigma_1$-complete while the latter are NP-complete. Nevertheless, it turns out that a precise relation between these two formalisms can be established if we impose the following linear restriction on lengths of derivations in a DPO grammar.
\begin{definition}\label{def_lin_DPO}
	\emph{A linearly restricted DPO grammar (lin-DPO grammar)} $Gr\restriction_c$ is a DPO grammar $Gr=\langle \mathcal{N},\mathcal{T},\mathcal{P}, Z\rangle$ equipped with an integer $c \in \mathbb{N}$. The language $L(Gr \restriction_c)$ generated by $Gr$ is the set of all hypergraphs $H \in \mathcal{H}(\mathcal{T})$ such that $Z \Rightarrow^k H$ for $k \le c |E_H|$.
\end{definition}
\begin{definition}
	\emph{An edgeful lin-DPO grammar} is a lin-DPO grammar \\ $\langle \mathcal{N},\mathcal{T},\mathcal{P}, Z\rangle$ such that $|E_Z|>0$ and $|E_L|>0$, $|E_R|>0$ for each $(L\mid R) \in \mathcal{P}$.
\end{definition}
Thus we limit the number of rule applications by a linear function w.r.t. the number of hyperedges of the resulting hypergraph. We have not find a similar definition in the field of graph grammars, but similar \textit{time-bounded grammars} have been studied in the string case (in particular, those with the linear bound) in \cite{Book71, Gladkii64}. In \cite{Book71}, the author notes (Proposition 2.5) that the linear restriction is the least possible one that does not lead to a collapse (i.e. to the case where only finite languages can be generated). This holds in the hypergraph case as well.

Our main goal is to prove the following result:
\begin{theorem}\label{th_main}
\leavevmode
\begin{enumerate}
	\item $\HL$-grammars are equivalent to edgeful lin-DPO grammars.
	\item The classes of languages without edgeless hypergraphs that are generated by $\HL^\ast$-grammars and by lin-DPO grammars are equal.
\end{enumerate}
\end{theorem}
The proof of both statements has two parts: from lin-DPO grammars to $\HL$-grammars (Section \ref{ssec_from_lindpo_to_hl}) and vice versa (Section \ref{ssec_from_hl_to_lindpo}). We describe the constructions only for $\HL^\ast$, but in Section \ref{ssec_rem_and_cor} we explain why they also work for $\HL$-grammars without any changes.

The first part involves encoding DPO rules within $\HL^\ast$, and a linear restriction naturally arises when we need to use this encoding at the level of $\HL^\ast$-grammars. This part was already presented in \cite{Pshenitsyn22_1} but we would like to explain the idea in this paper as well to make it self-contained; moreover we improve the presentation.
The second part is more straightforward: we need to describe axioms and rules of $\HL$ using DPO rules and then to show that the linear restriction is satisfied. There are, however, technical issues concerned with disconnected antecedents, which make the final construction more complicated. In this paper, we present the construction and sketch the proof of its correctness.

\subsection{From Lin-DPO Grammars to HL-grammars}\label{ssec_from_lindpo_to_hl}

In \cite{Pshenitsyn22_1}, an observation is made: an application of a rule of a DPO grammar, which consists of an inverse replacement followed by a straightforward replacement, can be modeled within $\HL$ using the rules ${\scriptstyle(\times L)}$ and ${\scriptstyle(\div L)}$. Namely, a rule $r=(L|R)$ is transformed into the type $\DPO(r) = \times(L) \div (R+\$_0^\bullet)$.

\begin{example}\label{ex_type_der_from_DPO_to_HL}
	Given the rule $\rho = (L_\rho | R_\rho)$ from Example \ref{ex_dpo_rule}, $\DPO(\rho)$ equals 
	$$\DPO(\rho) = \times\left( \vcenter{\hbox{{\tikz[baseline=.1ex]{
					\node[node, label=above:{\scriptsize $(1)$}] (N) {};
					\node[node, below left=4mm and 4mm of N, label=left:{\scriptsize $(2)$}] (N1) {};
					\node[node, below right=4mm and 4mm of N, label=right:{\scriptsize $(3)$}] (N2) {};
					\draw[->] (N) -- node[above left] {$l$} (N1);
					\draw[->] (N) -- node[above right] {$r$} (N2);
	}}}} \right) \div \left( \vcenter{\hbox{{\tikz[baseline=.1ex]{
					\node[node] (N) {};
					\node[node, below = 2.7mm of N, label=left:{\scriptsize $(2)$}] (N1) {};
					\node[node, right=11mm of N1, label=right:{\scriptsize $(3)$}] (N2) {};
					\node[node, right=11mm of N, label=right:{\scriptsize $(1)$}] (N3) {};
					\node[hyperedge, right = 2.5mm of N] (F) {$f$};
					\draw[->] (N1) -- node[below] {$t$} (N2);
					\draw[-] (N) -- node[above] {\scriptsize 1} (F);
	}}}}\;\;\vcenter{\hbox{{\tikz[baseline=.1ex]{\node[hyperedge] {$\$_0$};}}}} \right).$$

	Compare the application (\ref{ex_application_dpo_rule}) of the rule $\rho$ from Example \ref{ex_dpo_rule} with the derivation (\ref{ex_der_HL_hyper}) in $\HL$ involving $\DPO(\rho) = \times(L_\rho)\div (R_\rho + \$_0^\bullet)$:
	
	\begin{equation}\label{ex_der_HL_hyper}
		\infer[{\scriptstyle (\div L)}]{
			\vcenter{\hbox{{\tikz[baseline=.1ex]{
							\node[node] (N) {};
							\node[node, below=4.2mm of N] (N2) {};
							\node[node, left=4mm and 8mm of N2, label=left:{\scriptsize $(1)$}] (N1) {};
							\node[hyperedge,left=12.5mm of N] (E2) {$f$};
							\node[node, left=2mm of E2] (N3) {};
							\draw[-] (N3) -- node[below] {\scriptsize 1} (E2);
							\draw[->] (N1) to[bend right = 50] node[below] {$t$} (N2);
							\draw[->] (N2) to[bend right = 0] node[above] {$q$} (N1);
			}}}}\;\;\vcenter{\hbox{{\tikz[baseline=.1ex]{\node[hyperedge] {$\DPO(\rho)$};}}}}\quad \to \quad
			\times\left(\vcenter{\hbox{{\tikz[baseline=.1ex]{
							\node[node] (N) {};
							\node[node, below left = 3.3mm and 3.3mm of N, label=left:{\scriptsize $(1)$}] (N1) {};
							\node[node, below right = 3.3mm and 3.3mm of N] (N2) {};
							\draw[->] (N2) -- node[below] {$q$} (N1);
							\draw[->] (N) to[bend right = 0] node[above] {$l$} (N1);
							\draw[->] (N) to[bend left = 0] node[above] {$r$} (N2);
			}}}}\right)
		}
		{
			\infer[{\scriptstyle (\times L)}]{
				\vcenter{\hbox{{\tikz[baseline=.1ex]{
								\node[node] (N) {};
								\node[hyperedge,below=2.5mm of N] (E1) {$\times(L_\rho)$};
								\node[node, below left=10mm and 6.3mm of N, label=left:{\scriptsize $(1)$}] (N1) {};
								\node[node, below right=10mm and 6.3mm of N] (N2) {};
								\draw[->] (N2) -- node[below] {$q$} (N1);
								\draw[-] (N) -- node[left] {\scriptsize 1} (E1);
								\draw[-] (N1) -- node[above left] {\scriptsize 2} (E1);
								\draw[-] (N2) -- node[above right] {\scriptsize 3} (E1);
				}}}} \to \times\left(\vcenter{\hbox{{\tikz[baseline=.1ex]{
					\node[node] (N) {};
					\node[node, below left = 3.3mm and 3.3mm of N, label=left:{\scriptsize $(1)$}] (N1) {};
					\node[node, below right = 3.3mm and 3.3mm of N] (N2) {};
					\draw[->] (N2) -- node[below] {$q$} (N1);
					\draw[->] (N) to[bend right = 0] node[above] {$l$} (N1);
					\draw[->] (N) to[bend left = 0] node[above] {$r$} (N2);
				}}}}\right)
			}
			{
					\vcenter{\hbox{{\tikz[baseline=.1ex]{
									\node[node] (N) {};
									\node[node, below left = 3.3mm and 3.3mm of N, label=left:{\scriptsize $(1)$}] (N1) {};
									\node[node, below right = 3.3mm and 3.3mm of N] (N2) {};
									\draw[->] (N2) -- node[below] {$q$} (N1);
									\draw[->] (N) to[bend right = 0] node[above] {$l$} (N1);
									\draw[->] (N) to[bend left = 0] node[above] {$r$} (N2);
					}}}} \to \times\left(\vcenter{\hbox{{\tikz[baseline=.1ex]{
					\node[node] (N) {};
					\node[node, below left = 3.3mm and 3.3mm of N, label=left:{\scriptsize $(1)$}] (N1) {};
					\node[node, below right = 3.3mm and 3.3mm of N] (N2) {};
					\draw[->] (N2) -- node[below] {$q$} (N1);
					\draw[->] (N) to[bend right = 0] node[above] {$l$} (N1);
					\draw[->] (N) to[bend left = 0] node[above] {$r$} (N2);
				}}}}\right)
			}
			&
			t^\bullet \to t
			&
			f^\bullet \to f
		}
	\end{equation}
	In (\ref{ex_der_HL_hyper}), we firstly apply the rule ${\scriptstyle(\times L)}$ thus performing an inverse replacement in the antecedent and then the rule ${\scriptstyle(\div L)}$ thus performing several replacements (in this example, some of them change nothing). After the application of ${\scriptstyle(\div L)}$ a hyperedge of rank $0$ labeled by $\DPO(\rho)$ appears (it arises from $\$_0^\bullet$). 
\end{example}
Informally, the $\DPO(\rho)$-labeled hyperedge is the trace of the application of $\rho$, the resource needed to apply it within $\HL$. As a consequence, if we consider a derivation $Z \Rightarrow_{r_1} G_1 \Rightarrow_{r_2} \dotsc \Rightarrow_{r_k} G_k=G$, then we can remodel it within $\HL$ starting with the sequent $Z \to \times(Z)$ and obtaining the sequent $G + (\DPO(r_1))^\bullet + \dotsc + (\DPO(r_k))^\bullet \to \times(Z)$; here we treat nonterminal and terminal symbols as primitive types. However, this is not enough to present a transformation of DPO grammars into $\HL$-grammars since we need to describe how to assign types of $\HL$ to symbols of the alphabet. So, our main objective is, given a lin-DPO grammar $\langle \mathcal{N},\mathcal{T},\mathcal{P},Z\rangle \restriction_c$, to define an $\HL$-grammar $HLG$ generating the same language.
To recall, a hypergraph $G$ is generated by an $\HL$-grammar $HLG = \langle \mathcal{T},S,\triangleright\rangle$ if there is a relabeling $f_G$ of $G$ by types of $\HL$ such that $lab_G(e) \triangleright f_G(e)$ and $\HL \vdash f_G(G) \to S$. The above reasonings suggest us to let $S \eqdef \times(Z)$. Now, let us compare the sequent $G + (\DPO(r_1))^\bullet + \dotsc + (\DPO(r_k))^\bullet \to \times(Z)$ with the sequent $f_G(G) \to S$. We need to somehow ``hide'' the floating hyperedges labeled by $\DPO(r_i)$ using $f_G$. We can do this by merging them with hyperedges of $G$ using the rule ${\scriptstyle(\times L)}$, which leads us to the following construction:
\begin{definition}
	Let $Gr=\langle \mathcal{N}, \mathcal{T}, \mathcal{P}, Z \rangle \restriction_c$ be a lin-DPO grammar. Then, assuming that symbols from $\mathcal{N}\cup \mathcal{T}$ are primitive types, we construct an \\ $\HL^\ast$-grammar $\mathrm{HLG}(Gr) = \langle \mathcal{T}, \times(Z), \triangleright \rangle$ where $\triangleright$ consists of the following pairs: 
	\begin{equation}\label{eqn_def_hlg_c}
		a \triangleright \times\left(a^\bullet + \sum_{r \in \mathcal{P}} k_{r}\cdot \DPO(r)^\bullet\right)
		\mbox{ for }
		a \in \mathcal{T}, \; k_r \in \mathbb{N}, \; \sum_{r \in \mathcal{P}} k_{r} \le c.
	\end{equation}
\end{definition}
Note that $\triangleright$ is a finite relation since there are finitely many combinations of $k_r \in \mathbb{N}$ satisfying the above requirements. Informally, we take an $a$-labeled hyperedge and add several $\DPO(r)$-labeled hyperedges to it. Here the linear restriction arises: the relation $\triangleright$ must be finite so we cannot let (\ref{eqn_def_hlg_c}) hold for all $k_r \in \mathbb{N}$; then it is natural to assume that types assigned to some symbol are able to store at most some fixed number $c$ of $\DPO(r)$-labeled hyperedges. Thus the total number of such hyperedges (and, consequently, the number of rule applications in a derivation of $G$ in a DPO grammar) is limited by $c|E_G|$.

The following result is proved in \cite[Theorem 1]{Pshenitsyn22_1}:
\begin{theorem}\label{th_from_dpo_to_hlg}
	$L(\mathrm{HLG}(Gr)) = L(Gr)$ for each lin-DPO grammar $Gr$.
\end{theorem}

\subsection{From HL-Grammars to Lin-DPO Grammars}\label{ssec_from_hl_to_lindpo}
The main idea of the second part is to introduce DPO rules modelling the inference rules of $\HL$. Before doing this, however, we must somehow represent sequents as hypergraphs. Let us start with showing how to model $\LC$ using the DPO approach; after that we will consider the general case.

Recall that a sequent $A_1,\dotsc,A_n \to B$ of the Lambek calculus is embedded in $\HL$ using string graphs: $$\vcenter{\hbox{{\tikz[baseline=.1ex]{
				\node[node,label=left:{\scriptsize $(1)$}] (N1) {};
				\node[node,right=8mm of N1] (N2) {};
				\node[right=3mm of N2] (A) {};
				\node[right=3mm of A] (B) {};
				\node[node,right=3mm of B] (N3) {};
				\node[node, right=8mm of N3,label=right:{\scriptsize $(2)$}] (N4) {};
				\draw[->] (N1) -- node[above] {$\tau(A_1)$} (N2);
				\draw[-] (N2) --  (A);
				\draw[->] (B) -- (N3);
				\draw[->,white] (A) -- node[above] {\color{black} $\dotsc$} (B);
				\draw[->] (N3) -- node[above] {$\tau(A_n)$} (N4);
}}}} \to \tau(B).$$
To combine the antecedent and the succedent into a single graph, let us distinguish types in antecedents and in succedents. For each $T \in Tp$ we introduce nonterminal symbols $T^\ominus$ (antecedental) and $T^\oplus$ (succedental). Then the above sequent is to be represented as follows:
$$
\vcenter{\hbox{{\tikz[baseline=.1ex]{
				\node[node] (N1) {};
				\node[node,right=10mm of N1] (N2) {};
				\node[right=3mm of N2] (A) {};
				\node[right=5mm of A] (B) {};
				\node[node,right=3mm of B] (N3) {};
				\node[node, right=10mm of N3] (N4) {};
				\draw[->] (N1) -- node[above] {$\tau(A_1)^\ominus$} (N2);
				\draw[-] (N2) --  (A);
				\draw[->] (B) -- (N3);
				\draw[->,white] (A) -- node[above] {\color{black} $\dotsc$} (B);
				\draw[->] (N3) -- node[above] {$\tau(A_n)^\ominus$} (N4);
				\draw[->] (N1) to[bend right=40] node[above] {$\tau(B)^\oplus$} (N4);
}}}}.
$$
In particular, the axiom $A \to A$ is converted into the hypergraph $\vcenter{\hbox{{\tikz[baseline=.1ex]{
				\node[node] (N1) {};
				\node[node, right=10mm of N1] (N2) {};
				\draw[->] (N1) -- node[above] {$\tau(A)^\ominus$} (N2);
				\draw[->] (N1) to[bend right=45] node[below] {$\tau(A)^\oplus$} (N2);
}}}}$

Following this representation, the rules of the Lambek calculus can be transformed into the following DPO rules (we omit the rules for $\backslash$ for brevity):
\begin{center}
	\begin{tabular}{cl}
		$
		\infer[]{\Gamma, B / A, \Pi, \Delta \to C}{\Gamma, B, \Delta \to C & \Pi \to A}
		$
		&
		$
		r^{B/A}_L=
		\left(
		{\color{red}
			\vcenter{\hbox{{\tikz[baseline=.1ex]{
							\node[node,label=below:{\scriptsize $(1)$}] (N1) {};
							\node[node,right=10mm of N1,label=below:{\scriptsize $(2)$}] (N2) {};
							\node[node,right=4mm of N2,label=below:{\scriptsize $(3)$}] (N3) {};
							\node[node, right=10mm of N3,label=below:{\scriptsize $(4)$}] (N4) {};
							\draw[->] (N1) -- node[above] {$\tau(B)^\ominus$} (N2);
							\draw[->] (N3) -- node[above] {$\tau(A)^\oplus$} (N4);
			}}}}
		}
		\middle|
		\vcenter{\hbox{{\tikz[baseline=.1ex]{
						\node[node,label=below:{\scriptsize $(1)$}] (N1) {};
						\node[node,right=13mm of N1,label=below:{\scriptsize $(3)$}] (N2) {};
						\node[node, right=5.5mm of N2,label=below:{\scriptsize $(4)$},label=above:{\scriptsize $(2)$}] (N3) {};
						\draw[->] (N1) -- node[above] {$\tau(B/A)^\ominus$} (N2);
		}}}}
		\right)
		$
		\\
		$
		\infer[]{\Pi \to B / A}{\Pi, A \to B}
		$
		&
		$
		r^{B/A}_R=\left(
		{\color{purple}
		\vcenter{\hbox{{\tikz[baseline=.1ex]{
						\node[node,label=below:{\scriptsize $(1)$}] (N1) {};
						\node[node,right=12mm of N1] (N2) {};
						\node[node, right=12mm of N2,label=below:{\scriptsize $(2)$}] (N3) {};
						\draw[->] (N1) -- node[above] {$\tau(B)^\oplus$} (N2);
						\draw[<-] (N2) -- node[above] {$\tau(A)^\ominus$} (N3);
		}}}}
		}
		\middle|
		\vcenter{\hbox{{\tikz[baseline=.1ex]{
						\node[node,label=below:{\scriptsize $(1)$}] (N1) {};
						\node[node,right=15mm of N1,label=below:{\scriptsize $(2)$}] (N2) {};
						\draw[->] (N1) -- node[above] {$\tau(B/A)^\oplus$} (N2);
		}}}}
		\right)
		$
		\\
	%
		$
		\infer[]{\Gamma, A \cdot B, \Delta \to C}{\Gamma, A, B, \Delta \to C}
		$
		&
		$
		r^{A\cdot B}_L=\left(
		{\color{violet}
		\vcenter{\hbox{{\tikz[baseline=.1ex]{
						\node[node,label=below:{\scriptsize $(1)$}] (N1) {};
						\node[node,right=12mm of N1] (N2) {};
						\node[node, right=12mm of N2,label=below:{\scriptsize $(2)$}] (N3) {};
						\draw[->] (N1) -- node[above] {$\tau(A)^\ominus$} (N2);
						\draw[->] (N2) -- node[above] {$\tau(B)^\ominus$} (N3);
		}}}}
		}
		\middle|
		\vcenter{\hbox{{\tikz[baseline=.1ex]{
						\node[node,label=below:{\scriptsize $(1)$}] (N1) {};
						\node[node,right=15mm of N1,label=below:{\scriptsize $(2)$}] (N2) {};
						\draw[->] (N1) -- node[above] {$\tau(A\cdot B)^\ominus$} (N2);
		}}}}
		\right)
		$
		\\
		%
		$
		\infer[]{\Pi, \Psi \to A \cdot B}{\Pi \to A & \Psi \to B}
		$
		&
		$
		r^{A\cdot B}_R=\left(
		{\color{blue}
		\vcenter{\hbox{{\tikz[baseline=.1ex]{
						\node[node,label=below:{\scriptsize $(1)$}] (N1) {};
						\node[node,right=10mm of N1,label=below:{\scriptsize $(2)$}] (N2) {};
						\node[node,right=4mm of N2,label=below:{\scriptsize $(3)$}] (N3) {};
						\node[node, right=10mm of N3,label=below:{\scriptsize $(4)$}] (N4) {};
						\draw[->] (N1) -- node[above] {$\tau(A)^\oplus$} (N2);
						\draw[->] (N3) -- node[above] {$\tau(B)^\oplus$} (N4);
		}}}}
		}
		\middle|
		\vcenter{\hbox{{\tikz[baseline=.1ex]{
						\node[node,label=below:{\scriptsize $(1)$}] (N1) {};
						\node[node,right=13mm of N1,label=below:{\scriptsize $(4)$}] (N2) {};
						\node[node, right=5.5mm of N2,label=below:{\scriptsize $(3)$},label=above:{\scriptsize $(2)$}] (N3) {};
						\draw[->] (N1) -- node[above] {$\tau(A \cdot B)^\oplus$} (N2);
		}}}}
		\right)
		$
		\\
	\end{tabular}
\end{center}
\begin{example}\label{ex_from_LC_to_DPO}
	The following derivation starts with the disjoint union of several axioms, and it models the derivation from Example \ref{ex_der_L_2}:
	$$
	\vcenter{\hbox{{\tikz[baseline=.1ex]{
					\node[node] (N1) {};
					\node[node, right=10mm of N1] (N2) {};
					\draw[->] (N1) -- node[above] {$p^\ominus$} (N2);
					\draw[->,blue] (N1) to[bend right=35] node[below] {$p^\oplus$} (N2);
	}}}}
	\quad
	\vcenter{\hbox{{\tikz[baseline=.1ex]{
					\node[node] (N1) {};
					\node[node, right=10mm of N1] (N2) {};
					\draw[->] (N1) -- node[above] {$r^\ominus$} (N2);
					\draw[->,blue] (N1) to[bend right=35] node[below] {$r^\oplus$} (N2);
	}}}}
	\quad
	\vcenter{\hbox{{\tikz[baseline=.1ex]{
					\node[node] (N1) {};
					\node[node, right=10mm of N1] (N2) {};
					\draw[->] (N1) -- node[above] {$q^\ominus$} (N2);
					\draw[->] (N1) to[bend right=35] node[below] {$q^\oplus$} (N2);
	}}}}
	\;\Rightarrow_{r^{p\cdot r}_R}\;
	\vcenter{\hbox{{\tikz[baseline=.1ex]{
					\node[node] (N1) {};
					\node[node,right=10mm of N1] (N2) {};
					\node[node,right=10mm of N2] (N3) {};
					\draw[->,red] (N1) -- node[above] {$p^\ominus$} (N2);
					\draw[->] (N2) -- node[above] {$r^\ominus$} (N3);
					\draw[->] (N1) to[bend right=60] node[above] {$\tau(p\cdot r)^\oplus$} (N3);
	}}}}
	\quad
	\vcenter
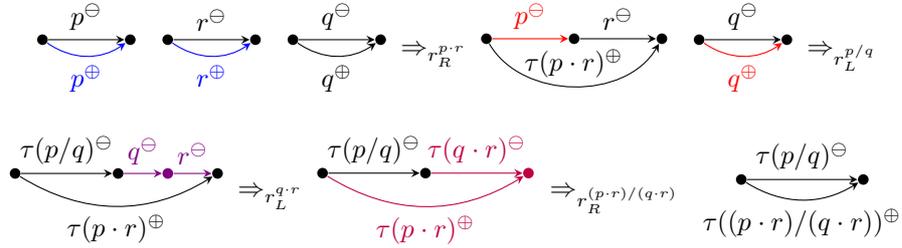
{\hbox{{\tikz[baseline=.1ex]{
					\node[node] (N1) {};
					\node[node, right=10mm of N1] (N2) {};
					\draw[->] (N1) -- node[above] {$q^\ominus$} (N2);
					\draw[->,red] (N1) to[bend right=35] node[below] {$q^\oplus$} (N2);
	}}}}
	\;\Rightarrow_{r^{p/q}_L}\;
	$$
	$$
	\vcenter{\hbox{{\tikz[baseline=.1ex]{
					\node[node] (N1) {};
					\node[node,right=12mm of N1] (N2) {};
					\node[node,right=5mm of N2,color=violet] (N3) {};
					\node[node,right=5mm of N3] (N4) {};
					\draw[->] (N1) -- node[above] {$\tau(p/q)^\ominus$} (N2);
					\draw[->,violet] (N2) -- node[above] {$q^\ominus$} (N3);
					\draw[->,violet] (N3) -- node[above] {$r^\ominus$} (N4);
					\draw[->] (N1) to[bend right=32] node[below] {$\tau(p\cdot r)^\oplus$} (N4);
	}}}}
	\;\Rightarrow_{r^{q\cdot r}_L}\;
	\vcenter{\hbox{{\tikz[baseline=.1ex]{
					\node[node] (N1) {};
					\node[node,right=12mm of N1] (N2) {};
					\node[node,right=12mm of N2,color=purple] (N3) {};
					\draw[->] (N1) -- node[above] {$\tau(p/q)^\ominus$} (N2);
					\draw[->,purple] (N2) -- node[above] {$\tau(q\cdot r)^\ominus$} (N3);
					\draw[->,purple] (N1) to[bend right=32] node[below] {$\tau(p\cdot r)^\oplus$} (N3);
	}}}}
	\;\Rightarrow_{r^{(p\cdot r)/(q\cdot r)}_R}\;
	\vcenter{\hbox{{\tikz[baseline=.1ex]{
					\node[node] (N1) {};
					\node[node, right=14.5mm of N1] (N2) {};
					\draw[->] (N1) -- node[above] {$\tau(p/q)^\ominus$} (N2);
					\draw[->] (N1) to[bend right=32] node[below] {$\tau((p\cdot r)/(q\cdot r))^\oplus$} (N2);
	}}}}
	$$
\end{example}

\begin{remark}\label{rem_problem_disconnected}
	In general, a problem of the above representation of sequents arises if antecedents of sequents are disconnected. Indeed, if we model an $\HL$ derivation in the way shown above, then all the sequents participating in the derivation including intermediate ones ``float'' in the same space. Now, if antecedents of several sequents have more than one connected component, then one might confuse, which one belongs to which sequent. For example, if we tried to generalize the above representation to some other cases, then the axiom sequent $p \to p$ would be represented as follows for $\rk(p)=0$: 
	$
	\vcenter{\hbox{\tikz[baseline=.1ex]{
		\node[hyperedge] (E) {$p^\ominus$};
	}}}
	\;\; 
	\vcenter{\hbox{\tikz[baseline=.1ex]{
		\node[hyperedge] (E) {$p^\oplus$};
	}}}
	$.
	Now consider two axiom sequents $p \to p$ and $q \to q$ in this representation together and assume that a rule similar to $r^{q/q}_L$ is applied:
	$$
	\vcenter{\hbox{\tikz[baseline=.1ex]{
				\node[hyperedge] (E) {$p^\ominus$};
	}}}
	\;\;
	\vcenter{\hbox{\tikz[baseline=.1ex]{
				\node[hyperedge] (E) {$p^\oplus$};
	}}}
	\;\;
	\vcenter{\hbox{\tikz[baseline=.1ex]{
				\node[hyperedge] (E) {$q^\ominus$};
	}}}
	\;\;
	\vcenter{\hbox{\tikz[baseline=.1ex]{
				\node[hyperedge] (E) {$q^\oplus$};
	}}}
	\;\Rightarrow\;
	\vcenter{\hbox{\tikz[baseline=.1ex]{
				\node[hyperedge] (E) {$p^\ominus$};
	}}}
	\;\;
	\vcenter{\hbox{\tikz[baseline=.1ex]{
				\node[hyperedge] (E) {$p^\oplus$};
	}}}
	\;\;
	\vcenter{\hbox{\tikz[baseline=.1ex]{
				\node[hyperedge] (E) {$\,\tau^\prime(q/q)^\ominus\,$};
	}}}
	$$
	Here $\tau^\prime(q/q)^\ominus$ should be defined as $q\div\left(\vcenter{\hbox{\tikz[baseline=.1ex]{
			\node[hyperedge] (E1) {$\$_0$};
			\node[hyperedge,right=2mm of E1] (E2) {$q$};
	}}}\right)$. But, finally, if we try to construct a sequent from the resulting hypergraph, then we obtain a non-derivable one: 
	$
	\vcenter{\hbox{\tikz[baseline=.1ex]{
		\node[hyperedge] (E) {$p$};
	}}}
	\;\;
	\vcenter{\hbox{\tikz[baseline=.1ex]{
		\node[hyperedge] (E) {$\,\tau^\prime(q/q)\,$};
	}}}
	\;\to\;
	p
	$. To fix this, we have to ``tie'' all the connected components belonging to the same sequent and to keep control of them.
\end{remark}

The general construction encounters the above remark. From now on, we assume that we are given an $\HL^\ast$-grammar $Gr=\langle \mathcal{T}, S, \triangleright\rangle$ and we construct a lin-DPO grammar $\mathrm{LDPO}(Gr)=\langle \mathcal{N},\mathcal{T},\mathcal{P},Z\rangle\restriction_C$ equivalent to it. If the readers would like to skip the construction description, then they can go to Lemma \ref{lemma_main}.

\textbf{Nonterminal symbols $\mathcal{N}$.} Let $Tp_{Gr} \eqdef \{T \mid \exists a \in \mathcal{T}: a \triangleright T\} \cup \{S\}$; types from this set and their subtypes are involved in derivations in $Gr$, so let us consider the set $ST_{Gr}$ of all subtypes of types from $Tp_{Gr}$. Then, let:
\begin{equation*}
\mathcal{N} = \{T^\ominus,T^\oplus \mid T \in ST_{Gr}, \rk(T^\ominus)=\rk(T^\oplus)=\rk(T)+1\} \sqcup \{\beta,\kappa,\varphi,q_1,q_2,q_3\}
\end{equation*}
Here $\beta,\kappa,\varphi,q_1,q_2,q_3$ are fresh nonterminal symbols such that $\rk(\beta)=1$, $\rk(\kappa)=2$, $\rk(\varphi) = \rk(S)$, and $\rk(q_i)=0$ for $i=1,2,3$. These symbols are used to keep control of derivation processes in the DPO grammar we construct. Namely, in the construction of $\mathrm{LDPO}(Gr)$ we are going to:
\begin{enumerate}
	\item distinguish between antecedental and succedental types using two kinds of labels $T^\ominus$ and $T^\oplus$ as in Example \ref{ex_from_LC_to_DPO};
	\item add an additional attachment node $v_\beta$ to each hyperedge of a sequent, which is represented as a hypergraph; in other words, we add a ``tentacle'' to each hyperedge and attach all these tentacles to $v_\beta$ (this allows one to keep all connected components of a sequent together);
	\item add binary edges going from all nodes of a sequent to $v_\beta$ labeled by $\kappa$ (``connector'') (this is important to control isolated nodes);
	\item add a unary hyperedge attached to $v_\beta$ and labeled by $\beta$ ($\beta$ stands for ``bouquet''; we need it to avoid cases described in Remark \ref{rem_problem_disconnected}).
	\item divide the derivation into 3 stages and use $q_i^\bullet$ to denote the $i$-th stage.
\end{enumerate} 

\textbf{The start hypergraph $Z$:} $Z=\varphi^\bullet + q_1^\bullet$. Here $\varphi^\bullet$ is the finishing hyperedge, it will be used at the last stage of a derivation.

\textbf{Rules $\mathcal{P}$.} 
In graph theory, it is much easier to explain by drawing rather than by writing. This is the reason why formal definitions of the rules of $\mathrm{LDPO}(Gr)$ are cumbersome. After presenting them we immediately proceed with examples.
\\
\textbf{First group $\mathcal{P}_1$.} At the first stage, axioms of $\HL$ represented as hypergraphs enter the game: if $A \in ST_{Gr}$ and $\rk(A)=n$, then $r^A_{Ax} \eqdef (q_1^\bullet | q_1^\bullet+Ax^A) \in \mathcal{P}_1$ where $Ax^A$ is the \textit{axiom hypergraph} defined as follows:
\begin{itemize}
	\item $V_{Ax^A}=\left\{ v_1,\dotsc,v_{n},v_\beta \right\}$; \qquad $E_{Ax^A}=\left\{ e^v_i \mid 1 \le i  \le n \right\} \sqcup \{e^\ominus,e^\oplus,e_\beta\}$;
	\item $att_{Ax^A}(e^v_i) = v_iv_\beta$, $att_{Ax^A}(e^\ominus)=att_{Ax^A}(e^\oplus)=v_1\dotsc v_{n} v_\beta$, $att_{Ax^A}(e_\beta) = v_\beta$; 
	\item $lab_{Ax^A}(e^v_i) = \kappa$, \; $lab_{Ax^A}(e^\ominus)=A^\ominus$, \; $lab_{Ax^A}(e^\oplus)=A^\oplus$, \; $lab_{Ax^A}(e_\beta) = \beta$; 
	\item $ext_{Ax^A}=\Lambda$.
\end{itemize}
\noindent
\textbf{Second group $\mathcal{P}_2$.}
At the second stage, the derivation in $\HL^\ast$ is modeled using DPO rules. First of all, let us define a hypergraph $G=Sep(T_1^{\circledcirc_1},\dotsc,T_m^{\circledcirc_m})$ where $T_i \in ST_{Gr}$ and $\circledcirc_i \in \{\oplus, \ominus\}$ as follows:
\begin{itemize}
	\item $V_{G}=\bigsqcup\limits_{i=1}^m \left\{ v^i_1,\dotsc,v^i_{\rk(T_i)},v^i_\beta \right\}$; \quad $E_{G}=\bigsqcup\limits_{i=1}^m \left\{ e^i_1,\dotsc,e^i_{\rk(T_i)},e^i_\beta,e^i \right\}$;
	\item $att_{G}(e^i_j) = v^i_j v^i_\beta$, \quad $att_{G}(e^i_\beta)=v^i_\beta$, \quad $att_{G}(e^i)=v^i_1\dotsc v^i_{\rk(T_i)} v^i_\beta$;
	\item $lab_{G}(e^i_j) = \kappa$, \quad $lab_{G}(e^i_\beta)=\beta$, \quad $lab_{G}(e^i)=T_i^{\circledcirc_i}$;
	\item $ext_{G} = att_{G}(e^1)\dotsc att_{G}(e^m)$.
\end{itemize}
For the sake of uniformity, $Sep()$ is the empty hypergraph.

Now we are ready to introduce four kinds of rules corresponding to the inference rules of $\HL$. Let $T=N \div D$ belong to $ST_{Gr}$ and $E_D=\{e^\$_D,d_1,\dotsc,d_k\}$. 
\\
\textbf{\textit{The rule $r^{T}_L$:}} $r^T_L=(q_2^\bullet+G^T_L | q_2^\bullet+H^T_L) \in \mathcal{P}_2$ where 
\begin{itemize}
	\item $G^T_L = Sep(N^\ominus,lab_D(d_1)^\oplus,\dotsc,lab_D(d_k)^\oplus)$.
	\item $V_{H^T_L} = V_D \sqcup \{v_\beta\}$; \quad $E_{H^T_L} = \left\{ e_v \mid v \in V_D \right\} \sqcup \left\{ e , e_\beta\right\}$; 
	\item $att_{H^T_L}(e_v) = vv_\beta$,\quad $att_{H^T_L}(e) = att_D(e^\$_D) v_\beta$, \quad $att_{H^T_L}(e_\beta) = v_\beta$; 
	\item $lab_{H^T_L}(e_v) = \kappa$, \quad $lab_{H^T_L}(e) = (N \div D)^\ominus$, \quad $lab_{H^T_L}(e_\beta) = \beta$; 
	\item $ext_{H^T_L} = ext_D v_\beta att_{D}(d_1)v_\beta\dotsc att_{D}(d_k) v_\beta$.
\end{itemize}

\noindent
\textbf{\textit{The rule $r^{T}_R$:}} $r^T_R=(q_2^\bullet+G^T_R | q_2^\bullet+H^T_R) \in \mathcal{P}_2$ where 
\begin{itemize}
	\item $V_{G^T_R} = V_D \sqcup \{v_\beta\}$; \quad $E_{G^T_R} = E_D \sqcup \left\{ e_v \mid v \in V_D \right\} \sqcup \left\{ e_N, e_\beta \right\}$; 
	\item $att_{G^T_R}(d_i) = att_D(d_i) v_\beta$, \quad $att_{G^T_R}(e_N) = ext_D v_\beta$, \quad $att_{G^T_R}(e_v) = vv_\beta$, \\ $att_{G^T_R}(e_\beta) = v_\beta$; \qquad $lab_{G^T_R}(d_i) = lab_D(d_i)^\ominus$, \quad $lab_{G^T_R}(e_N) = N^\oplus$, \\ $lab_{G^T_R}(e_v) = \kappa$, \quad $lab_{G^T_R}(e_\beta) = \beta$; \qquad $ext_{G^T_R} = lab_D(e^\$_D)v_\beta$.
	\item $H^T_R = Sep((N\div D)^\oplus)$.
\end{itemize}
\noindent
Let $U= \times(M)$ belong to $ST_{Gr}$ and let $E_M=\{m_1,\dotsc,m_l\}$. 
\\
\textbf{\textit{The rule $r^{U}_L$:}} $r^U_L=(q_2^\bullet+G^U_L | q_2^\bullet+H^U_L) \in \mathcal{P}_2$ where 
\begin{itemize}
	\item $V_{G^U_L} = V_M \sqcup \{v_\beta\}$; \quad $E_{G^U_L} = E_M \sqcup \left\{ e_v \mid v \in V_M \right\} \sqcup \left\{ e_\beta \right\}$; 
	\item $att_{G^U_L}(m_i) = att_M(m_i) v_\beta$, \quad $att_{G^U_L}(e_v) = vv_\beta$, \quad $att_{G^U_L}(e_\beta) = v_\beta$; 
	\item $lab_{G^U_L}(m_i) = lab_M(m_i)^\ominus$, $lab_{G^U_L}(e_v) = \kappa$, $lab_{G^U_L}(e_\beta) = \beta$; $ext_{G^U_L} = ext_M v_\beta$.
	\item $H^U_L = Sep((\times(M))^\ominus)$.
\end{itemize}\noindent
\textbf{\textit{The rule $r^{U}_R$:}} $r^U_R=(q_2^\bullet+G^U_R | q_2^\bullet+H^U_R) \in \mathcal{P}_2$ where 
\begin{itemize}
	\item $G^U_R = Sep(lab_M(m_1)^\oplus,\dotsc,lab_M(m_l)^\oplus)$.
	\item $V_{H^U_R} = V_M \sqcup \{v_\beta\}$; \quad $E_{H^U_R} = \left\{ e_v \mid v \in V_M \right\} \sqcup \left\{e, e_\beta\right\}$; 
	\item $att_{H^U_R}(e_v) = vv_\beta$, \quad $att_{H^U_R}(e) = ext_M v_\beta$, \quad $att_{H^U_R}(e_\beta) = v_\beta$; 
	\item $lab_{H^U_R}(e_v) = \kappa$, \quad $lab_{H^U_R}(e) = (\times(M))^\oplus$, \quad $lab_{H^U_R}(e_\beta) = \beta$; 
	\item $ext_{H^U_R} = att_{M}(m_1)v_\beta\dotsc att_{M}(m_l) v_\beta$.
\end{itemize}
\noindent
Finally, the set $\mathcal{P}$ consists of the following rules:
\begin{enumerate}
	\item Rules from $\mathcal{P}_1$ and $\mathcal{P}_2$;
	\item $(q_1^\bullet|q_{2}^\bullet)$, $(q_2^\bullet|q_{3}^\bullet)$,
	\quad
	$\left(q_3^\bullet+\kappa^\bullet\middle|q_3^\bullet+\left(\vcenter{\hbox{{\tikz[baseline=.1ex]{
					\node[node,label=left:{\scriptsize $(1)$}] (N1) {};
					\node[node, right=3mm of N1,label=right:{\scriptsize $(2)$}] (N2) {};
	}}}}\right)\right)$;
	\item $(q_3^\bullet + (A^\ominus)^\bullet| q_3^\bullet + a^{\bullet_1})$ where $a \in \mathcal{T}$ such that $a \triangleright A$, and $a^{\bullet_1}$ is obtained from $a^\bullet$ by adding an isolated node as the $(\rk(a)+1)$-st external node;
	\item $(q_3^\bullet + F_L| F_R)$ where
	\begin{itemize}
		\item $V_{F_L} = \{v_1,\dotsc,v_{2\rk(S)},v_\beta\}$; \quad $E_{F_L} = \{e_1,e_2,e_\beta\}$;
		\item $att_{F_L}(e_1)=v_1\dotsc v_{\rk(S)}$, $att_{F_L}(e_2)=v_{\rk(S)+1}\dotsc v_{2\rk(S)}v_\beta$;
		\item $lab_{F_L}(e_1)=\varphi$, $lab_{F_L}(e_2)=S^\oplus$; $ext_{F_L} = att_{F_L}(e_1)att_{F_L}(e_2)$.
		\item $V_{F_R} = \{u_1,\dotsc,u_{\rk(S)}\}$; $E_{F_L} = \emptyset$; $ext_{F_R} = u_1\dotsc u_{\rk(S)} u_1\dotsc u_{\rk(S)}$.
	\end{itemize}
\end{enumerate}
To choose a constant $C$, let us define \emph{the full size $||T||$} of a type $T$ as follows: 
\begin{enumerate}
	\item $||p||=\rk(p)+1$ for $p \in Pr$;
	\item $||N \div D|| = ||N||+\sum_{e \in E_D}||lab_D(e)||+|V_D|+\rk(e^\$_D)+1$;
	\item $||\times(M)|| = \sum_{e \in E_M}||lab_M(e)||+|V_M|+|ext_M|+1$. 
\end{enumerate}
Informally, $||T||$ encounters the total number of primitive types, connectives and nodes within $T$. Let $k \eqdef \max_{T \in Tp_{Gr}} ||T||$. The main observation is as follows:
\begin{lemma}\label{lemma_size_der}
	If $\HL^\ast \vdash H \to S$, then the total number of axioms and rule applications in its derivation as well as the total number of nodes in $H$ does not exceed $||lab_H(h_1)||+\dotsc+||lab_H(h_m)|| + ||S|| \le k\left(|E_H|+1\right)$ where $E_H=\{h_1,\dotsc,h_m\}$. 
\end{lemma} 
It is straightforwardly proved by induction on the length of a derivation: we just need to check that this holds for axioms and that each inference rule preserves this property.
Finally, let $C=2k+3$. This completes the definition of $\mathrm{LDPO}(Gr) = \langle \mathcal{N}, \mathcal{T}, \mathcal{P}, Z\rangle \restriction_C$. 

\begin{example}
	If $A_0,A_2$ are types such that $\rk(A_0)=0$, $\rk(A_2)=2$, then:
	$$
	r^{A_0}_{Ax} = 
	\left(\vcenter{\hbox{{\tikz[baseline=.1ex]{\node[hyperedge] {$q_1$};}}}}
	\;\;\middle|\;\;
	\vcenter{\hbox{{\tikz[baseline=.1ex]{
					\node[node] (N) {};
					\node[hyperedge,left=3mm of N] (E1) {$A_0^\ominus$};
					\node[hyperedge,right=3mm of N] (E2) {$A_0^\oplus$};
					\node[hyperedge,below=3mm of N] (E3) {$\beta$};
					\draw[-] (N) -- (E1);
					\draw[-] (N) -- (E2);
					\draw[-] (N) -- (E3);
	}}}}
	\quad
	\vcenter{\hbox{{\tikz[baseline=.1ex]{\node[hyperedge] {$q_1$};}}}}
	\right)
	;
	\quad
	r^{A_2}_{Ax} = 
	\left(\vcenter{\hbox{{\tikz[baseline=.1ex]{\node[hyperedge] {$q_1$};}}}}
	\;\middle|\;
	\vcenter{\hbox{{\tikz[baseline=.1ex]{
					\node[] (O) {};
					\node[node,left=12mm of O] (N1) {};
					\node[node,right=12mm of O] (N2) {};
					\node[hyperedge,above right=-1mm and 5.5mm of N1] (E1) {$A_2^\ominus$};
					\node[hyperedge,below left=-1mm and 5.5mm of N2] (E2) {$A_2^\oplus$};
					\node[node,below=6mm of O] (N3) {};
					\node[hyperedge,below=2.2mm of N3] (E3) {$\beta$};
					\draw[-] (N1) to[bend left=13] node[above] {\scriptsize 1} (E1);
					\draw[-] (E1) to[bend left=13] node[above right] {\scriptsize 2} (N2);
					\draw[-] (N2) to[bend left=13] node[below] {\scriptsize 2} (E2);
					\draw[-] (E2) to[bend left=13] node[below left] {\scriptsize 1} (N1);
					\draw[-] (E3) -- (N3);
					\draw[-] (E1) to[bend right=20] node[right] {\scriptsize 3} (N3);
					\draw[-] (E2) to[bend left=10] node[right] {\scriptsize 3} (N3);
					\draw[->] (N1) to[out=-90,in=180] node[below] {$\kappa$} (N3);
					\draw[->] (N2) to[out=-90,in=0] node[below] {$\kappa$} (N3);
	}}}}
	\;
	\vcenter{\hbox{{\tikz[baseline=.1ex]{\node[hyperedge] {$q_1$};}}}}
	\right).
	$$
	Hereinafter we omit the subscript {\scriptsize 1} if a line connects a hyperedge of rank 1 and its only attachment node.
\end{example}

\begin{example}
	Let $T_\rho=\times(L_\rho)$ where $L_\rho$ is the left-hand side of the rule $\rho$ from Example \ref{ex_dpo_rule} (see also Example \ref{ex_type_der_from_DPO_to_HL}). Then:
	$$
	r^{T_\rho}_L = \left(
	\vcenter{\hbox{{\tikz[baseline=.1ex]{
					\node[node,label=above:{\scriptsize $(1)$}] (O1) {};
					\node[below=3.7mm of O1] (O2) {};
					\node[node,below=3.7mm of O2, label=below right :{\scriptsize $(4)$}] (N) {};
					\node[node,left=12mm of N, label=left:{\scriptsize $(2)$}] (N1) {};
					\node[node,right=12mm of N, label=right:{\scriptsize $(3)$}] (N2) {};
					\node[hyperedge,left=3.2mm of O2] (E1) {$l^\ominus$};
					\node[hyperedge,right=3.2mm of O2] (E2) {$r^\ominus$};
					\node[hyperedge,below =3.8mm of N] (E3) {$\beta$};
					\draw[-] (O1) -- node[left] {\scriptsize 1} (E1);
					\draw[-] (O1) -- node[right] {\scriptsize 1} (E2);
					\draw[-] (E1) -- node[left] {\scriptsize 2} (N1);
					\draw[-] (E2) -- node[right] {\scriptsize 2} (N2);
					\draw[-] (E1) -- node[left] {\scriptsize 3} (N);
					\draw[-] (E2) -- node[right] {\scriptsize 3} (N);
					\draw[->] (O1) -- node[left] {$\kappa$} (N);
					\draw[->] (N1) -- node[below left] {$\kappa$} (N);
					\draw[->] (N2) -- node[below right] {$\kappa$} (N);
					\draw[-] (E3) -- (N);
	}}}}
	\;
	\vcenter{\hbox{{\tikz[baseline=.1ex]{\node[hyperedge] {$q_2$};}}}}
	\;\;
	\middle|
	\vcenter{\hbox{{\tikz[baseline=.1ex]{
					\node[hyperedge] (E1) {$T_\rho^\ominus$};
					\node[node,above=4mm of E1, label=right:{\scriptsize $(1)$}] (N1) {};
					\node[node,left=6mm of E1, label=left:{\scriptsize $(2)$}] (N2) {};
					\node[node,right=6mm of E1, label=right:{\scriptsize $(3)$}] (N3) {};
					\node[node,below=4mm of E1, label=right:{\scriptsize $(4)$}] (N4) {};
					\node[hyperedge,below =3mm of N4] (E2) {$\beta$};
					\draw[-] (E1) -- node[left] {\scriptsize 1} (N1);
					\draw[-] (E1) -- node[above] {\scriptsize 2} (N2);
					\draw[-] (E1) -- node[above] {\scriptsize 3} (N3);
					\draw[-] (E1) -- node[right] {\scriptsize 4} (N4);
					\draw[-] (N4) -- (E2);
					\draw[->] (N2) -- node[left] {$\kappa$} (N4);
					\draw[->] (N3) -- node[right] {$\kappa$} (N4);
					\draw[->] (N1) 
					to[out=180,in=90] (-17mm,0mm) node[left] {$\kappa$}
					to[out=270,in=180] (N4)
					;
	}}}}
	\;
	\vcenter{\hbox{{\tikz[baseline=.1ex]{\node[hyperedge] {$q_2$};}}}}
	\right)
	$$
	$$
	r^{T_\rho}_R = \left(
	\vcenter{\hbox{{\tikz[baseline=.1ex]{
					\node[hyperedge] (E1) {$l^\oplus$};
					\node[node,left=6mm of E1, label=above:{\scriptsize $(1)$}] (N2) {};
					\node[node,right=6mm of E1, label=above:{\scriptsize $(2)$}] (N3) {};
					\node[node,below=4mm of E1, label=right:{\scriptsize $(3)$}] (N4) {};
					\node[hyperedge,below =3mm of N4] (E2) {$\beta$};
					\draw[-] (E1) -- node[above] {\scriptsize 1} (N2);
					\draw[-] (E1) -- node[above] {\scriptsize 2} (N3);
					\draw[-] (E1) -- node[right] {\scriptsize 3} (N4);
					\draw[-] (N4) -- (E2);
					\draw[->] (N2) -- node[left] {$\kappa$} (N4);
					\draw[->] (N3) -- node[right] {$\kappa$} (N4);
	}}}}
	\;
	\vcenter{\hbox{{\tikz[baseline=.1ex]{
					\node[hyperedge] (E1) {$r^\oplus$};
					\node[node,left=6mm of E1, label=above:{\scriptsize $(4)$}] (N2) {};
					\node[node,right=6mm of E1, label=above:{\scriptsize $(5)$}] (N3) {};
					\node[node,below=4mm of E1, label=right:{\scriptsize $(6)$}] (N4) {};
					\node[hyperedge,below =3mm of N4] (E2) {$\beta$};
					\draw[-] (E1) -- node[above] {\scriptsize 1} (N2);
					\draw[-] (E1) -- node[above] {\scriptsize 2} (N3);
					\draw[-] (E1) -- node[right] {\scriptsize 3} (N4);
					\draw[-] (N4) -- (E2);
					\draw[->] (N2) -- node[left] {$\kappa$} (N4);
					\draw[->] (N3) -- node[right] {$\kappa$} (N4);
	}}}}
	\;
	\vcenter{\hbox{{\tikz[baseline=.1ex]{\node[hyperedge] {$q_2$};}}}}
	\;\;
	\middle|
	\vcenter{\hbox{{\tikz[baseline=.1ex]{
					\node[hyperedge] (E1) {$T_\rho^\oplus$};
					\node[node,above=3.5mm of E1, label=above:{\scriptsize $(1)\,(4)$}] (N1) {};
					\node[node,left=6mm of E1, label=above:{\scriptsize $(2)$}] (N2) {};
					\node[node,right=6mm of E1, label=above:{\scriptsize $(5)$}] (N3) {};
					\node[node,below=5mm of E1, label=right:{\scriptsize $(6)$}, label=below left:{\scriptsize $(3)$}] (N4) {};
					\node[hyperedge,below =3.5mm of N4] (E2) {$\beta$};
					\draw[-] (E1) -- node[right] {\scriptsize 1} (N1);
					\draw[-] (E1) -- node[above] {\scriptsize 2} (N2);
					\draw[-] (E1) -- node[above] {\scriptsize 3} (N3);
					\draw[-] (E1) -- node[right] {\scriptsize 4} (N4);
					\draw[-] (N4) -- (E2);
					\draw[->] (N2) -- node[left] {$\kappa$} (N4);
					\draw[->] (N3) -- node[right] {$\kappa$} (N4);
					\draw[->] (N1) 
					to[out=180,in=90] (-14mm,0mm) node[left] {$\kappa$}
					to[out=270,in=180] (N4)
					;
	}}}}
	\;
	\vcenter{\hbox{{\tikz[baseline=.1ex]{\node[hyperedge] {$q_2$};}}}}
	\right)
	$$
\end{example}

\begin{example}\label{ex_translation_HL_to_DPO}
	Let $U_\rho=\DPO(\rho) = T_\rho \div (R_\rho+\$_0^\bullet)$ from Example \ref{ex_type_der_from_DPO_to_HL}. Then:
	$$
	r^{U_\rho}_L = \left(
	\vcenter{\hbox{{\tikz[baseline=.1ex]{
					\node[hyperedge] (E1) {$T_\rho^\ominus$};
					\node[node,above=3.5mm of E1, label=right:{\scriptsize $(1)$}] (N1) {};
					\node[node,left=5mm of E1, label=above:{\scriptsize $(2)$}] (N2) {};
					\node[node,right=5mm of E1, label=above:{\scriptsize $(3)$}] (N3) {};
					\node[node,below=5mm of E1, label=right:{\scriptsize $(4)$}] (N4) {};
					\node[hyperedge,below =3mm of N4] (E2) {$\beta$};
					\draw[-] (E1) -- node[left] {\scriptsize 1} (N1);
					\draw[-] (E1) -- node[above] {\scriptsize 2} (N2);
					\draw[-] (E1) -- node[above] {\scriptsize 3} (N3);
					\draw[-] (E1) -- node[right] {\scriptsize 4} (N4);
					\draw[-] (N4) -- (E2);
					\draw[->] (N2) -- node[left] {$\kappa$} (N4);
					\draw[->] (N3) -- node[right] {$\kappa$} (N4);
					\draw[->] (N1) 
					to[out=180,in=90] (-13mm,0mm) node[left] {$\kappa$}
					to[out=270,in=180] (N4)
					;
					\node[hyperedge, right=16mm of E1] (E21) {$f^\oplus$};
					\node[node,left=4mm of E21, label=above:{\scriptsize $(5)$}] (N22) {};
					\node[node,below=5mm of E21, label=left:{\scriptsize $(6)$}] (N24) {};
					\node[hyperedge,below =3mm of N24] (E22) {$\beta$};
					\draw[-] (E21) -- node[above] {\scriptsize 1} (N22);
					\draw[-] (E21) -- node[right] {\scriptsize 2} (N24);
					\draw[-] (N24) -- (E22);
					\draw[->] (N22) -- node[left] {$\kappa$} (N24);
					\node[hyperedge, right=9mm of E21] (E31) {$t^\oplus$};
					\node[node,left=4mm of E31, label=above:{\scriptsize $(7)$}] (N32) {};
					\node[node,right=4mm of E31, label=above:{\scriptsize $(8)$}] (N33) {};
					\node[node,below=5mm of E31, label=right:{\scriptsize $(9)$}] (N34) {};
					\node[hyperedge,below =3mm of N34] (E32) {$\beta$};
					\draw[-] (E31) -- node[above] {\scriptsize 1} (N32);
					\draw[-] (E31) -- node[above] {\scriptsize 2} (N33);
					\draw[-] (E31) -- node[right] {\scriptsize 3} (N34);
					\draw[-] (N34) -- (E32);
					\draw[->] (N32) -- node[left] {$\kappa$} (N34);
					\draw[->] (N33) -- node[right] {$\kappa$} (N34);
	}}}}
	\vcenter{\hbox{{\tikz[baseline=.1ex]{\node[hyperedge] {$q_2$};}}}}
	\;
	\middle|
	\vcenter{\hbox{{\tikz[baseline=.1ex]{
					\node[] (O) {};
					\node[node, above left=6mm and 7mm of O, label=left:{\scriptsize $(5)$}] (N1) {};
					\node[node, above right=6mm and 7mm of O, label=right:{\scriptsize $(1)$}] (N2) {};
					\node[node, below left=6mm and 7mm of O, label=below:{\scriptsize $(7)$}, label=left:{\scriptsize $(2)$}] (N3) {};
					\node[node, below right=6mm and 7mm of O, label=below:{\scriptsize $(8)$}, label=right:{\scriptsize $(3)$}] (N4) {};
					\node[hyperedge,right=6.5mm of O] (E1) {$U_\rho^\ominus$};
					\node[node, label=above:{\scriptsize $(4)$}, label=below:{\scriptsize $(6)$}, label=right:{\scriptsize $(9)$}] (N5) {};
					\node[hyperedge,left =6mm of N5] (E2) {$\beta$};
					\node[hyperedge,right=1mm of E1] (E3) {$q_2$};
					\draw[->] (N1) to[bend right=15] node[above] {$\kappa$} (N5);
					\draw[->] (N2) to[bend right=15] node[above] {$\kappa$} (N5);
					\draw[->] (N3) to[bend right=15] node[below] {$\kappa$} (N5);
					\draw[->] (N4) to[bend right=15] node[below] {$\kappa$} (N5);
					\draw[-] (E1) to[bend right=26] (N5);
					\draw[-] (N5) -- (E2);
	}}}}
	\right)
	$$
	$$
	r^{U_\rho}_R = \left(
	\vcenter{\hbox{{\tikz[baseline=.1ex]{
					\node[] (O) {};
					\node[node, above left=6mm and 12mm of O] (N1) {};
					\node[node, above right=6mm and 12mm of O] (N2) {};
					\node[node, below left=6mm and 12mm of O] (N3) {};
					\node[node, below right=6mm and 12mm of O] (N4) {};
					\node[hyperedge,right=10.45mm of N1] (E1) {$f^\ominus$};
					\node[node] (N5) {};
					\node[above right =-4mm and -2.4mm of O, label= left:{\scriptsize $(1)$}] {};
					\node[hyperedge,below=4mm of O] (E2) {$t^\ominus$};
					\node[hyperedge,left =10mm of N5] (E3) {$\beta$};
					\node[hyperedge,right =10mm of O] (E4) {$T_\rho^\oplus$};
					\node[hyperedge,right=18mm of O] (E5) {$q_2$};
					\draw[->] (N1) to[bend right=15] node[above] {$\kappa$} (N5);
					\draw[->] (N2) to[bend right=15] node[above] {$\kappa$} (N5);
					\draw[->] (N3) to[bend right=15] node[right] {$\kappa$} (N5);
					\draw[->] (N4) to[bend right=15] node[below] {$\kappa$} (N5);
					\draw[-] (N5) to[bend left=0] (E3);
					\draw[-] (N1) to[bend left=0] node[above] {\scriptsize 1} (E1);
					\draw[-] (E1) to[bend left=0] node[left] {\scriptsize 2} (N5);
					\draw[-] (N3) to[bend left=0] node[below] {\scriptsize 1} (E2);
					\draw[-] (N4) to[bend left=0] node[below right] {\scriptsize 2} (E2);
					\draw[-] (E2) to[bend left=0] node[right] {\scriptsize 3} (N5);
					\draw[-] (N2) to[bend left=0] node[right] {\scriptsize 1} (E4);
					\draw[-] (N3) to[bend right=55] node[below left] {\scriptsize 2} (E4);
					\draw[-] (N4) to[bend left=0] node[right] {\scriptsize 3} (E4);
					\draw[-] (E4) to[bend left=0] node[above right] {\scriptsize 4} (N5);
	}}}}
	\;
	\middle|
	\vcenter{\hbox{{\tikz[baseline=.1ex]{
	\node[hyperedge] (E21) {$U_\rho^\oplus$};
	\node[node,below=3mm of E21, label=left:{\scriptsize $(1)$}] (N24) {};
	\node[hyperedge,below=3mm of N24] (E22) {$\beta$};
	\draw[-] (E21) -- (N24);
	\draw[-] (N24) -- (E22);
	}}}}
	\;\;
	\vcenter{\hbox{{\tikz[baseline=.1ex]{\node[hyperedge] {$q_2$};}}}}
	\right)
	$$
\end{example}

The following lemma is crucial:

\begin{lemma}\label{lemma_main}
	$L(\mathrm{LDPO}(Gr)) = L(Gr)$ if $L(Gr)$ includes no edgeless hypergraphs.
\end{lemma}
Its proof sketch can be found in Appendix \ref{app_proof_lemma_main}.

\subsection{Around Theorem \ref{th_main}}\label{ssec_rem_and_cor}
Let us start with the proof of Theorem \ref{th_main}:
\begin{proof}[of Theorem \ref{th_main}]
	The second statement directly follows from Theorem \ref{th_from_dpo_to_hlg} and Lemma \ref{lemma_main}: each lin-DPO grammar $Gr$ can be converted into an equivalent $\HL^\ast$-grammar $\mathrm{HLG}(Gr)$, and each $\HL^\ast$-grammar $Gr^\prime$ can be converted into an equivalent lin-DPO grammar $\mathrm{LDPO}(Gr^\prime)$. To prove the first statement note that if $Gr$ is edgeful, than $\mathrm{HLG}(Gr)$ can be considered as an $\HL$-grammar since its types belong to $\HL$ and any antecedent in any derivation of this grammar has hyperedges. Conversely, we would like to say that $\mathrm{LDPO}(Gr^\prime)$ is edgeful; however, the rule $(q_3^\bullet+F_L \mid F_R)$ violates this condition since $|E_{F_R}|=0$. To fix this, we replace this rule by the rules $(\varphi^\bullet + F^{a,A}_L + q_3^\bullet \mid F_R+a^\bullet)$ where:
	\begin{itemize}
		\item $a \in \mathcal{T}$ ia a symbol, $A \in \mathit{Tp}$ is a type and $a \triangleright A$;
		\item $V_{F^{a,A}_L} = \{v_i\}_{i=1}^{\rk(S)+\rk(a)}\sqcup\{v_\beta\}$; \quad $E_{F^{a,A}_L} = \{e_S,e_A,e_\beta\}$; \quad $att_{F^{a,A}_L}(e_\beta)=v_\beta$,
		\\
		$att_{F^{a,A}_L}(e_S)=v_{1}\dotsc v_{\rk(S)}v_\beta$, $att_{F^{a,A}_L}(e_A)=v_{\rk(S)+1}\dotsc v_{\rk(S)+\rk(A)}v_\beta$;
		\item $lab_{F^{a,A}_L}(e_S)=S^\oplus$, $lab_{F^{a,A}_L}(e_A)=A^\ominus$, $lab_{F^{a,A}_L}(e_\beta)=\beta$;
		\item $ext_{F^{a,A}_L} = v_1\dotsc v_{\rk(S)+\rk(A)}$.
	\end{itemize}
	This new rule combines the last rule of the original grammar with one of the rules of the form $(q_3^\bullet + (A^\ominus)^\bullet| q_3^\bullet + a^{\bullet_1})$. Thus $\mathrm{LDPO}(Gr^\prime)$ is edgeful. 
	\qed
\end{proof}

We can derive some nice corollaries from the constructions of grammars for free. E.g. consider an alternative definition of lin-DPO grammars where the linear restriction is the inner property of a grammar:
\begin{definition}
	A DPO grammar $HGr=\langle \mathcal{N},\mathcal{T},\mathcal{P}, Z\rangle$ satisfies the strong linear restriction property if there exists $c \in \mathbb{N}$ such that for all derivations $Z \Rightarrow^k H$ with $H \in \mathcal{H}(T)$ it holds that $k \le c|E_H|$.
\end{definition}

It follows from the proof of Lemma \ref{lemma_main} (see Appendix \ref{app_proof_lemma_main}) that, for each derivation $Z \Rightarrow^k H$ in $\mathrm{LDPO}(Gr)$, we have $k \le C|E_H|$, thus $\mathrm{LDPO}(Gr)$ satisfies the strong linear restriction property. Hence the following corollary holds:
\begin{corollary}\label{cor_slrp}
	If a lin-DPO grammar $Gr$ does not generate edgeless hypergraphs, then the equivalent grammar $\mathrm{LDPO}(\mathrm{HLG}(Gr))$ is edgeful and it satisfies the strong linear restriction property.
\end{corollary}

The definition of lin-DPO grammars is simpler than that of $\HL$-grammars, and from now on, thanks to Theorem \ref{th_main}, we can reason about the former if we want to study the expressive power of the latter. As an example, let us relate contextual hyperedge-replacement grammars defined in \cite{DrewesH15} with $\HL$-grammars. In \cite{DrewesH15}, a hypergraph is a tuple $G=\langle \dot{G},\bar{G},att_G,\dot{l}_G,\bar{l}_G\rangle$ where $\dot{G}$ are nodes, $\bar{G}$ are hyperedges, $att_G$ is the same as in our definition, $\dot{l}_G:\dot{G} \to \dot{\mathcal{C}}$ labels nodes, $\bar{l}_G:\bar{G} \to \bar{\mathcal{C}}$ labels hyperedges (see further notation in that article). From this hypergraph we can obtain a hypergraph in our sense if we replace node labels by additional hyperedges of rank 1 attached to these nodes: $t(G) \eqdef \langle \dot{G},\bar{G} \sqcup \dot{G},att^\prime_G,l_G, \Lambda\rangle$ where $att^\prime_G\vert_{\bar{G}}=att_G$, $att^\prime(v) = v$; $l_G\vert_{\dot{G}}=\dot{l}_G$, $l_G\vert_{\bar{G}}=\bar{l}_G$.
\begin{corollary}
	For each contextual grammar $CGr$ that does not generate the empty hypergraph there is an $\HL$-grammar generating $\{t(H) \mid H \in L(CGr)\}$.
\end{corollary}
\begin{proof}[sketch]
	It is proved in \cite[Theorem 3.19]{DrewesH15} that one can eliminate empty and chain rules in any contextual grammar. Assume that $CGr = \langle \mathcal{C},\mathcal{R},Z\rangle$ satisfies this requirement (note that $|E_Z|>0$). We construct an equivalent DPO grammar $DGr=\langle X, \dot{\mathcal{C}}\sqcup \bar{\mathcal{C}},\mathcal{P},t(Z)\rangle$: for each $(L,R) \in \mathcal{R}$ let \\$L^\prime = \langle V_{t(L)},E_{t(L)},att_{t(L)},lab_{t(L)},v_1\dotsc v_l\rangle$ (here $V_{t(L)} = \{v_1,\dotsc,v_l\}$) and $R^\prime = \langle V_{t(R)},E_{t(R)},att_{t(R)},lab_{t(R)},v_1\dotsc v_l\rangle$; then we add the rule $(L^\prime \mid R^\prime)$ to $\mathcal{P}$. It is not hard to prove that each derivation $G \Rightarrow^\ast H$ in $CGr$ can be straightforwardly transformed into a derivation $t(G) \Rightarrow^\ast t(H)$ in $DGr$ and vice versa.
	
	In the proof of \cite[Theorem 4.1]{DrewesH15}, it is mentioned that the length of each derivation is linear in the size of the graph generated. Formally, if $Z \Rightarrow^k G$, then $k \le c(|\dot{G}|+|\bar{G}|) = c|E_{t(G)}|$ for some $c$. Thus for each derivation $t(Z) \Rightarrow^k t(G)$ it holds that $k \le c|E_{t(G)}|$, hence $DGr\restriction_c$ is an edgeful lin-DPO grammar equivalent to $CGr$. It remains to use Theorem \ref{th_main}.
	\qed
\end{proof}

\section{Conclusion}\label{sec_conclusion}
The equivalence between $\HL$-grammars and linearly restricted DPO grammars builds a nice bridge between the logic-based and the rule-based approach providing an independent grammar characterization of linear time bounded DPO grammars. The hypergraph Lambek calculus is expected to be of interest for logicians and linguists, but from the formal point of view lin-DPO grammars are much more attractive since they have simple definitions based on the well-studied DPO approach. 

The linear restriction for DPO grammars is of interest on its own since time-bounded graph grammars have not been studied yet, to our best knowledge. While still being NP-complete, lin-DPO grammars are more powerful than hyperedge replacement grammars. In the extended version of this paper we would like to prove that the class of languages generated by them is closed under intersection, which would complete our previous results presented in \cite{Pshenitsyn22_2}; this would also contrast to \cite[Theorem 6.2]{Book71}, which states that languages generated by linearly restricted context-sensitive grammars are not closed under intersection. Studying more relations of this grammar formalism with other ones is interesting; in particular, is it the case that lin-DPO grammars generate all languages from NP? and what about string languages generated by them?

Although $\HL$-grammars are more sophisticated, they give us some new perspectives on the theory of graph grammars. In particular, the idea of logic-based and type-based graph grammars can go beyond the formalism considered in this work. In the paper that I am going to submit to the workshop GCM 2023 I would like to show how to relate $\HL$ with the first-order intuitionistic linear logic and then to suggest a general notion of type-logical hypergraph grammars over an arbitrary first-order logic.

\appendixpage
\appendixtocoff
\appendix

\section{Proof of Lemma \ref{lemma_main}}\label{app_proof_lemma_main}

\begin{proof}
	Firstly, it is important to notice that the set of rules is finite: indeed, they involve only types from $ST_{Gr}$, and the latter set is finite. Now, let us analyze derivations in $L(\mathrm{LDPO}(Gr))$. Firstly, we note that, if $Z \Rightarrow^\ast G$, then either $G$ contains exactly one $q_i$-labeled hyperedge or $G$ is terminal. Besides, if $Z \Rightarrow G_1 \Rightarrow \dotsc \Rightarrow G_n$, and $G_j$ contains the hyperedge $q_{i(j)}^\bullet$, then $i(j)$ is non-decreasing. Formally, this is proved by induction on $n$; we just need to observe that each rule takes exactly one hyperedge of the form $q_i^\bullet$ and adds exactly one such hyperedge as well. As a consequence, any derivation can be divided into eight stages:
	\begin{enumerate}
		\item Applications of rules of the first group;
		\\
		Replacement of $q_1^\bullet$ by $q_2^\bullet$;
		\item Applications of rules of the second group;
		\\
		Replacement of $q_2^\bullet$ by $q_3^\bullet$;
		\item Applications of $\left(q_3^\bullet+\kappa^\bullet\middle|q_3^\bullet+\left(\vcenter{\hbox{{\tikz[baseline=.1ex]{
						\node[node,label=left:{\scriptsize $(1)$}] (N1) {};
						\node[node, right=3mm of N1,label=right:{\scriptsize $(2)$}] (N2) {};
		}}}}\right)\right)$
		and
		$(q_3^\bullet + (A^\ominus)^\bullet| q_3^\bullet + a^{\bullet_1})$
		;
		\item An application of the rule $(q_3^\bullet + F_L| F_R)$.
	\end{enumerate}
	Let us now consider each stage. It is straightforward to prove by induction that any hypergraph appearing at Stage 1 before the replacement of $q_1^\bullet$ is of the form $\varphi^\bullet+q_1^\bullet+\sum_{i=1}^{m_1} Ax^{A_i}$ for some types $A_i$ and some $m_1\in \mathbb{N}$; thus Stage 1 introduces axiom hypergraphs corresponding to axioms $A_i^\bullet \to A_i$. At the end of Stage 1 we have a hypergraph $\varphi^\bullet+q_2^\bullet+\sum_{i=1}^{M_1} Ax^{A_i}$. The number of steps at this stage equals $M_1+1$ where $M_1$ is the total number of axioms introduced.
	
	To analyze Stage 2, let us define a \emph{sequent hypergraph corresponding to a sequent $H \to A$} as follows:
	\begin{definition}
		If $H \to A$ is a sequent, then $\mathrm{SH}(H \to A)$ is a hypergraph $\langle V_H \sqcup \{v_\beta\},E_H \sqcup V_H \sqcup \{e_A,e_\beta\},att,lab,\Lambda\rangle$ where
		\begin{itemize}
			\item $att(e) = att_H(e)v_\beta$ and $lab(e) = lab_H(e)^\ominus$ for $e \in E_H$; 
			\item $att(v) = vv_\beta$ and $lab(v) = \kappa$ for $v \in V_H$; 
			\item $att(e_A)=ext_Hv_\beta$ and $lab(e_A) = A^\oplus$; \item $att(e_\beta)=v_\beta$ and $lab(e_\beta) = \beta$.
		\end{itemize}
	\end{definition}
	We claim that at Stage 2 (before the replacement of $q_2^\bullet$) any hypergraph appearing at some step of a derivation must be of the form $\varphi^\bullet+q_2^\bullet+\sum_{i=1}^{m_2} \mathrm{SH}(H_i \to B_i)$ for some $m_2 \in \mathbb{N}$ and for some \emph{derivable} sequents $H_i \to B_i$; let us call such hypergraphs \emph{derivational}. The proof is by induction on the number of steps of the derivation at Stage 2. To prove the base case, it suffices to notice that $\varphi^\bullet+q_2^\bullet+\sum_{i=1}^{M_1} Ax^{A_i}$ is derivational since $Ax^{A_i} = \mathrm{SH}(A_i^\bullet \to A_i)$.
	
	To prove the induction step, it suffices to explain why any application of a rule from the second group preserves the property of being derivational. For example, let us consider the case where $r^T_L$ for $T=N \div D$ (we use the same notation as in the definitions of the rules of $\mathrm{LDPO}(Gr)$). To apply this rule, we must select an $N^\ominus$-labeled hyperedge and hyperedges labeled by $lab_D(d_i)^\oplus$ within a derivational hypergraph $\varphi^\bullet+q_2^\bullet+\sum_{i=1}^{m_2} \mathrm{SH}(H_i \to B_i)$. Note that no two of these hyperedges can be selected from the same sequent hypergraph $\mathrm{SH}(H_i \to B_i)$ because this would imply that the latter has at least two $\beta$-labeled hyperedges, which is not the case. Thus, without loss of generality we can assume that we select the $lab_D(d_i)^\oplus$-labeled hyperedge in $\mathrm{SH}(H_i \to B_i)$ for $i=1,\dotsc,k$ and the $N^\ominus$-labeled hyperedge in $\mathrm{SH}(H_{k+1} \to B_{k+1})$. Consequently, $B_i = lab_D(d_i)$. Then, we remove all these hyperdges, we remove the $\kappa$-labeled edges adjacent to them as well as $\beta$-labeled hyperedges. Finally, we combine $\mathrm{SH}(H_{1} \to B_{1})$, $\dotsc$, $\mathrm{SH}(H_{k+1} \to B_{k+1})$ together according to the right-hand side of $r^T_L$. This results in a hypergraph $\mathrm{SH}(H^\prime \to B^\prime)$ where $H^\prime \to B^\prime$ is obtained from $H_{k+1} \to B_{k+1}$ and $H_1 \to B_1$, $\dotsc$, $H_k \to B_k$ by applying ${\scriptstyle(\div L)}$: we just need to check that the rule application glues the initial sequent hypergraphs correctly and returns all necessary $\kappa$-labeled and $\beta$-labeled hyperedges.
	
	The remaining rules can be analyzed in a similar way as well but, unfortunately, the real detailed proof would require several pages of formal descriptions of intermediate hypergraphs and it would be extremely tedious. A better way is to look at the examples, in particular, at Example \ref{ex_app_proof_lemma_main}.
	
	To sum up Stage 2, the resulting hypergraph must be of the form $\varphi^\bullet+q_3^\bullet+\sum_{i=1}^{M_2} \mathrm{SH}(H_i \to B_i)$. It is crucial to notice that our rules allow to remove a $\beta$-labeled hyperedge only one time, thus there is only one $\beta$-labeled hyperedge in this hypergraph, hence $M_2 = 1$. Therefore, at the beginning of Stage 3, we have $\varphi^\bullet+q_3^\bullet+\mathrm{SH}(H_1 \to B_1)$. At Stage 3, we remove all $\kappa$-labeled hyperedges and we replace each $A^\ominus$-labeled one by a terminal $a$-labeled hyperedge for $a \triangleright A$; note that we remove the additional attachment node and thus disconnect a hyperedge from the node with the attached $\beta$-labeled hyperedge. At the end of Stage 3, we have the hypergraph $\varphi^\bullet+q_3^\bullet+H_1^\prime$ where 
	\begin{itemize}
		\item $V_{H_1^\prime} = V_{H_1} \sqcup \{v_\beta\}$;
		\item $E_{H_1^\prime} = E_{H_1}\sqcup \{e_{B_1},e_\beta\}$;
		\item $att_{H_1^\prime}(e) = att_{H_1}(e)$ for $e \in E_{H_1}$, $att_{H_1^\prime}(e_{B_1}) = ext_{H_1}v_\beta$, $att(e_\beta) = v_\beta$;
		\item $lab_{H_1^\prime}(e) \triangleright lab_{H_1}(e)$ for $e \in E_{H_1}$, $lab_{H_1^\prime}(e_{B_1}) = B_1^\oplus$, $lab_{H_1^\prime}(e_\beta)=\beta$;
		\item $ext_{H_1^\prime} = \Lambda$.
	\end{itemize}
	At Stage 4, we fuse the $i$-th attachment node of the $\varphi$-labeled hyperedge with the $i$-th attachment node of $e_{B_1}$ and we also remove $e_{B_1}$ and $e_\beta$ along with $v_\beta$. Thus we obtain the hypergraph $f(H_1)$ where $f(e) = lab_{H_1^\prime}(e)$. Note that it must be the case that $B_1=S$ according to the rule $(q_3^\bullet+F_L \mid F_R)$. Concluding, we know that $H_1 \to B_1 = H_1 \to S$ is derivable and that the resulting hypergraph is $f(H_1)$ where $f(e) \triangleright lab_{H_1}(e)$ for all $e \in E_{H_1}$. Thus we have proved that $L(\mathrm{LDPO}(Gr))\subseteq L(Gr)$.
	
	The converse inclusion is proved using the same reasoning: we just notice that any derivation of $\HL^\ast$ can be modeled according to the four-stage scheme described above. We only need to keep track of the number of rules applied in a derivation. Let us do this: assume that $H \in L(Gr)$, that is, there is a relabeling $f_H$ such that $lab_H(e) \triangleright f_H(e)$ and $\HL^\ast \vdash f_H(H) \to S$. To derive $H$ in $\mathrm{LDPO}(Gr)$, we need to remodel the derivation of $f_H(H) \to S$ at Stages 1 and 2 thus obtaining the hypergraph $\varphi^\bullet+q_3^\bullet+\mathrm{SH}(f_H(H) \to S)$ and then make $H$ from the latter hypergraph. Let us count the number of rule applications:
	\begin{enumerate}
		\item At Stage 1, the number of rule applications equals $K_1+1$ where $K_1$ is the number of axioms in the derivation of $f_H(H) \to S$;
		\item At Stage 2, the number of rule applications equals $K_2+1$ where $K_2$ is the number of inference rules in the derivation of $f_H(H) \to S$;
		\item At Stage 3, the number of rule applications equals $|V_H|+|E_H|$ (because there is exactly one $\kappa$-labeled hyperedge attached to each node of $f_H(H)$ in $\mathrm{SH}(f_H(H) \to S)$ at the beginning of Stage 3);
		\item At Stage 4, one rule is applied.
	\end{enumerate}
	The total number of rule applications equals $K_1+K_2+|V_H|+|E_H|+3$. According to Lemma \ref{lemma_size_der}, $K_1+K_2+|V_H| \le k \cdot (|E_H|+1)$. Thus $K_1+K_2+|V_H|+|E_H|+3 \le k \cdot (|E_H|+1)+3 \le (2k+3)|E_H|=C|E_H|$. It is important here that $H$ has hyperedges, so we can use the inequality $1 \le |E_H|$.
	\qed
\end{proof}

\begin{example}\label{ex_app_proof_lemma_main}
	The derivation presented below represents Stage 1 and Stage 2: during the first three steps we introduce axiom hypergraphs, and during the next four steps we apply rules corresponding to the inference rules of $\HL$. This derivation corresponds to that from Example \ref{ex_from_LC_to_DPO} (the difference between them is that here we add $\kappa$-labeled and $\beta$-labeled hyperedges and additional attachment nodes). After the derivation, we also present the rules used in it that correspond to the inference rules of $\HL$.
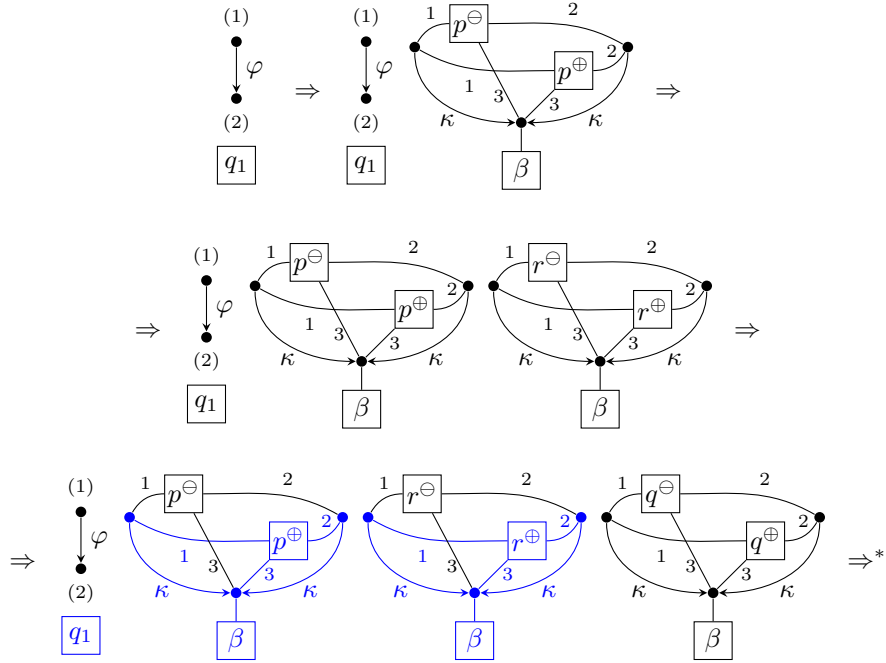
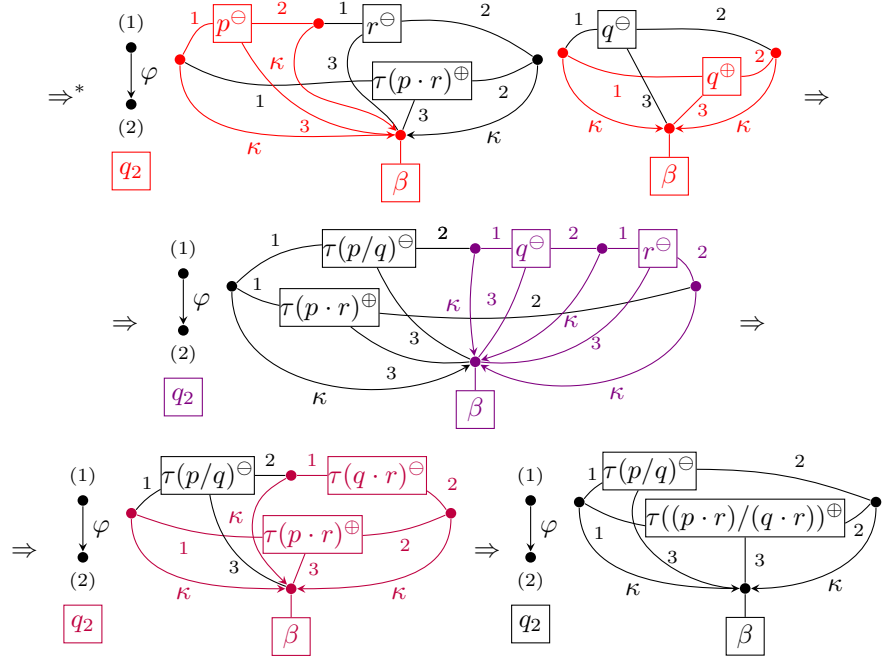
	$$
	\vcenter{\hbox{{\tikz[baseline=.1ex]{
					\node[node,label=above:{\scriptsize $(1)$}] (N1) {};
					\node[node, below=6mm of N1,label=below:{\scriptsize $(2)$}] (N2) {};
					\node[hyperedge, below=5.5mm of N2] (E) {$q_1$};
					\draw[->] (N1) -- node[right] {$\varphi$} (N2);
	}}}}
	\;\;\Rightarrow\;\;
	\vcenter{\hbox{{\tikz[baseline=.1ex]{
					\node[node,label=above:{\scriptsize $(1)$}] (N1) {};
					\node[node, below=6mm of N1,label=below:{\scriptsize $(2)$}] (N2) {};
					\node[hyperedge, below=5.5mm of N2] (E) {$q_1$};
					\draw[->] (N1) -- node[right] {$\varphi$} (N2);
	}}}}
	\;
	\vcenter{\hbox{{\tikz[baseline=.1ex]{
					\node[] (O) {};
					\node[node,left=12mm of O] (N1) {};
					\node[node,right=12mm of O] (N2) {};
					\node[hyperedge,above right=0mm and 4mm of N1] (E1) {$p^\ominus$};
					\node[hyperedge,below left=0mm and 4mm of N2] (E2) {$p^\oplus$};
					\node[node,below=8mm of O] (N3) {};
					\node[hyperedge,below=3mm of N3] (E3) {$\beta$};
					\draw[-] (N1) to[out=60,in=180] node[above] {\scriptsize 1} (E1);
					\draw[-] (E1) to[out=0,in=150] node[above right] {\scriptsize 2} (N2);
					\draw[-] (N2) to[out=240,in=0] node[above] {\scriptsize 2} (E2);
					\draw[-] (E2) to[out=180,in=-30] node[below left] {\scriptsize 1} (N1);
					\draw[-] (E3) -- (N3);
					\draw[-] (E1) -- node[below] {\scriptsize 3} (N3);
					\draw[-] (E2) -- node[right] {\scriptsize 3} (N3);
					\draw[->] (N1) to[out=-90,in=180] node[below] {$\kappa$} (N3);
					\draw[->] (N2) to[out=-90,in=0] node[below] {$\kappa$} (N3);
	}}}}
	\;\;\Rightarrow\;\;
	$$
	$$
	\Rightarrow\;\;
	\vcenter{\hbox{{\tikz[baseline=.1ex]{
					\node[node,label=above:{\scriptsize $(1)$}] (N1) {};
					\node[node, below=6mm of N1,label=below:{\scriptsize $(2)$}] (N2) {};
					\node[hyperedge, below=5.5mm of N2] (E) {$q_1$};
					\draw[->] (N1) -- node[right] {$\varphi$} (N2);
	}}}}
	\;
	\vcenter{\hbox{{\tikz[baseline=.1ex]{
					\node[] (O) {};
					\node[node,left=12mm of O] (N1) {};
					\node[node,right=12mm of O] (N2) {};
					\node[hyperedge,above right=0mm and 4mm of N1] (E1) {$p^\ominus$};
					\node[hyperedge,below left=0mm and 4mm of N2] (E2) {$p^\oplus$};
					\node[node,below=8mm of O] (N3) {};
					\node[hyperedge,below=3mm of N3] (E3) {$\beta$};
					\draw[-] (N1) to[out=60,in=180] node[above] {\scriptsize 1} (E1);
					\draw[-] (E1) to[out=0,in=150] node[above right] {\scriptsize 2} (N2);
					\draw[-] (N2) to[out=240,in=0] node[above] {\scriptsize 2} (E2);
					\draw[-] (E2) to[out=180,in=-30] node[below left] {\scriptsize 1} (N1);
					\draw[-] (E3) -- (N3);
					\draw[-] (E1) -- node[below] {\scriptsize 3} (N3);
					\draw[-] (E2) -- node[right] {\scriptsize 3} (N3);
					\draw[->] (N1) to[out=-90,in=180] node[below] {$\kappa$} (N3);
					\draw[->] (N2) to[out=-90,in=0] node[below] {$\kappa$} (N3);
	}}}}
	\;\;
	\vcenter{\hbox{{\tikz[baseline=.1ex]{
					\node[] (O) {};
					\node[node,left=12mm of O] (N1) {};
					\node[node,right=12mm of O] (N2) {};
					\node[hyperedge,above right=0mm and 4mm of N1] (E1) {$r^\ominus$};
					\node[hyperedge,below left=0mm and 4mm of N2] (E2) {$r^\oplus$};
					\node[node,below=8mm of O] (N3) {};
					\node[hyperedge,below=3mm of N3] (E3) {$\beta$};
					\draw[-] (N1) to[out=60,in=180] node[above] {\scriptsize 1} (E1);
					\draw[-] (E1) to[out=0,in=150] node[above right] {\scriptsize 2} (N2);
					\draw[-] (N2) to[out=240,in=0] node[above] {\scriptsize 2} (E2);
					\draw[-] (E2) to[out=180,in=-30] node[below left] {\scriptsize 1} (N1);
					\draw[-] (E3) -- (N3);
					\draw[-] (E1) -- node[below] {\scriptsize 3} (N3);
					\draw[-] (E2) -- node[right] {\scriptsize 3} (N3);
					\draw[->] (N1) to[out=-90,in=180] node[below] {$\kappa$} (N3);
					\draw[->] (N2) to[out=-90,in=0] node[below] {$\kappa$} (N3);
	}}}}
	\;\;\Rightarrow\;\;
	$$
	$$
	\Rightarrow\;\;
	\vcenter{\hbox{{\tikz[baseline=.1ex]{
					\node[node,label=above:{\scriptsize $(1)$}] (N1) {};
					\node[node, below=6mm of N1,label=below:{\scriptsize $(2)$}] (N2) {};
					\node[blue,hyperedge, below=5.5mm of N2] (E) {$q_1$};
					\draw[->] (N1) -- node[right] {$\varphi$} (N2);
	}}}}
	\;
	\vcenter{\hbox{{\tikz[baseline=.1ex]{
					\node[] (O) {};
					\node[blue,node,left=12mm of O] (N1) {};
					\node[blue,node,right=12mm of O] (N2) {};
					\node[hyperedge,above right=0mm and 4mm of N1] (E1) {$p^\ominus$};
					\node[blue,hyperedge,below left=0mm and 4mm of N2] (E2) {$p^\oplus$};
					\node[blue,node,below=8mm of O] (N3) {};
					\node[blue,hyperedge,below=3mm of N3] (E3) {$\beta$};
					\draw[-] (N1) to[out=60,in=180] node[above] {\scriptsize 1} (E1);
					\draw[-] (E1) to[out=0,in=150] node[above right] {\scriptsize 2} (N2);
					\draw[blue,-] (N2) to[out=240,in=0] node[above] {\scriptsize 2} (E2);
					\draw[blue,-] (E2) to[out=180,in=-30] node[below left] {\scriptsize 1} (N1);
					\draw[blue,-] (E3) -- (N3);
					\draw[-] (E1) -- node[below] {\scriptsize 3} (N3);
					\draw[blue,-] (E2) -- node[right] {\scriptsize 3} (N3);
					\draw[blue,->] (N1) to[out=-90,in=180] node[below] {$\kappa$} (N3);
					\draw[blue,->] (N2) to[out=-90,in=0] node[below] {$\kappa$} (N3);
	}}}}
	\;\;
	\vcenter{\hbox{{\tikz[baseline=.1ex]{
					\node[] (O) {};
					\node[blue,node,left=12mm of O] (N1) {};
					\node[blue,node,right=12mm of O] (N2) {};
					\node[hyperedge,above right=0mm and 4mm of N1] (E1) {$r^\ominus$};
					\node[blue,hyperedge,below left=0mm and 4mm of N2] (E2) {$r^\oplus$};
					\node[blue,node,below=8mm of O] (N3) {};
					\node[blue,hyperedge,below=3mm of N3] (E3) {$\beta$};
					\draw[-] (N1) to[out=60,in=180] node[above] {\scriptsize 1} (E1);
					\draw[-] (E1) to[out=0,in=150] node[above right] {\scriptsize 2} (N2);
					\draw[blue,-] (N2) to[out=240,in=0] node[above] {\scriptsize 2} (E2);
					\draw[blue,-] (E2) to[out=180,in=-30] node[below left] {\scriptsize 1} (N1);
					\draw[blue,-] (E3) -- (N3);
					\draw[-] (E1) -- node[below] {\scriptsize 3} (N3);
					\draw[blue,-] (E2) -- node[right] {\scriptsize 3} (N3);
					\draw[blue,->] (N1) to[out=-90,in=180] node[below] {$\kappa$} (N3);
					\draw[blue,->] (N2) to[out=-90,in=0] node[below] {$\kappa$} (N3);
	}}}}
	\;\;
	\vcenter{\hbox{{\tikz[baseline=.1ex]{
					\node[] (O) {};
					\node[node,left=12mm of O] (N1) {};
					\node[node,right=12mm of O] (N2) {};
					\node[hyperedge,above right=0mm and 4mm of N1] (E1) {$q^\ominus$};
					\node[hyperedge,below left=0mm and 4mm of N2] (E2) {$q^\oplus$};
					\node[node,below=8mm of O] (N3) {};
					\node[hyperedge,below=3mm of N3] (E3) {$\beta$};
					\draw[-] (N1) to[out=60,in=180] node[above] {\scriptsize 1} (E1);
					\draw[-] (E1) to[out=0,in=150] node[above right] {\scriptsize 2} (N2);
					\draw[-] (N2) to[out=240,in=0] node[above] {\scriptsize 2} (E2);
					\draw[-] (E2) to[out=180,in=-30] node[below left] {\scriptsize 1} (N1);
					\draw[-] (E3) -- (N3);
					\draw[-] (E1) -- node[below] {\scriptsize 3} (N3);
					\draw[-] (E2) -- node[right] {\scriptsize 3} (N3);
					\draw[->] (N1) to[out=-90,in=180] node[below] {$\kappa$} (N3);
					\draw[->] (N2) to[out=-90,in=0] node[below] {$\kappa$} (N3);
	}}}}
	\;\;\Rightarrow^\ast\;\;
	$$
	$$
	\Rightarrow^\ast\;\;
	\vcenter{\hbox{{\tikz[baseline=.1ex]{
					\node[node,label=above:{\scriptsize $(1)$}] (N1) {};
					\node[node, below=6mm of N1,label=below:{\scriptsize $(2)$}] (N2) {};
					\node[red,hyperedge, below=5.5mm of N2] (E) {$q_2$};
					\draw[->] (N1) -- node[right] {$\varphi$} (N2);
	}}}}
	\;
	\vcenter{\hbox{{\tikz[baseline=.1ex]{
					\node[] (O) {};
					\node[red,node,left=27mm of O] (N1) {};
					\node[node,right=16mm of O] (N2) {};
					\node[red,node,above left=3mm and 9mm of O] (N4) {};
					\node[red,hyperedge,left= 8mm of N4] (E1) {$p^\ominus$};
					\node[hyperedge,right=5mm of N4] (E4) {$r^\ominus$};
					\node[hyperedge,below right =-1mm and -5mm of O] (E2) {$\tau(p\cdot r)^\oplus$};
					\node[red,node,below=8mm of O] (N3) {};
					\node[red,hyperedge,below=3mm of N3] (E3) {$\beta$};
					\draw[red,-] (N1) to[out=60,in=180] node[above] {\scriptsize 1} (E1);
					\draw[red,-] (E1) -- node[above] {\scriptsize 2} (N4);
					\draw[-] (N4) -- node[above] {\scriptsize 1} (E4);
					\draw[-] (E4) to[bend left=20] node[above right] {\scriptsize 2} (N2);
					\draw[-] (N2) to[out=210,in=0] node[below] {\scriptsize 2} (E2);
					\draw[-] (E2) to[out=180,in=-30] node[below left] {\scriptsize 1} (N1);
					\draw[red,-] (E3) -- (N3);
					\draw[red,-] (E1) to[bend right=30] node[below] {\scriptsize 3} (N3);
					\draw[-] (E2) -- node[right] {\scriptsize 3} (N3);
					\draw[-] (E4) 
					to[out=-150,in=90] (-7mm,-0.5mm) node[left] {\scriptsize 3}
					to[out=-90,in=120] (N3)
					;
					\draw[red,->] (N1) to[out=-90,in=180] node[below] {$\kappa$} (N3);
					\draw[->] (N2) to[out=-90,in=0] node[below] {$\kappa$} (N3);
					\draw[red,->] (N4) 
					to[out=-150,in=90] (-14mm,-0.5mm) node[left] {$\kappa$}
					to[out=-90,in=135] (N3)
					;
	}}}}
	\;\;
	\vcenter{\hbox{{\tikz[baseline=.1ex]{
					\node[] (O) {};
					\node[red,node,left=12mm of O] (N1) {};
					\node[red,node,right=12mm of O] (N2) {};
					\node[hyperedge,above right=0mm and 4mm of N1] (E1) {$q^\ominus$};
					\node[red,hyperedge,below left=0mm and 4mm of N2] (E2) {$q^\oplus$};
					\node[red,node,below=8mm of O] (N3) {};
					\node[red,hyperedge,below=3mm of N3] (E3) {$\beta$};
					\draw[-] (N1) to[out=60,in=180] node[above] {\scriptsize 1} (E1);
					\draw[-] (E1) to[out=0,in=150] node[above right] {\scriptsize 2} (N2);
					\draw[red,-] (N2) to[out=240,in=0] node[above] {\scriptsize 2} (E2);
					\draw[red,-] (E2) to[out=180,in=-30] node[below left] {\scriptsize 1} (N1);
					\draw[red,-] (E3) -- (N3);
					\draw[-] (E1) -- node[below] {\scriptsize 3} (N3);
					\draw[red,-] (E2) -- node[right] {\scriptsize 3} (N3);
					\draw[red,->] (N1) to[out=-90,in=180] node[below] {$\kappa$} (N3);
					\draw[red,->] (N2) to[out=-90,in=0] node[below] {$\kappa$} (N3);
	}}}}
	\;\;\Rightarrow\;\;
	$$
	$$
	\Rightarrow\;\;
	\vcenter{\hbox{{\tikz[baseline=.1ex]{
					\node[node,label=above:{\scriptsize $(1)$}] (N1) {};
					\node[node, below=6mm of N1,label=below:{\scriptsize $(2)$}] (N2) {};
					\node[violet,hyperedge, below=5.5mm of N2] (E) {$q_2$};
					\draw[->] (N1) -- node[right] {$\varphi$} (N2);
	}}}}
	\;
	\vcenter{\hbox{{\tikz[baseline=.1ex]{
					\node[] (O) {};
					\node[node,left=30mm of O] (N1) {};
					\node[violet,node,right=27mm of O] (N2) {};
					\node[violet,node,above=3mm of O] (N4) {};
					\node[hyperedge,left= 7mm of N4] (E1) {$\tau(p/q)^\ominus$};
					\node[violet,hyperedge,right=4mm of N4] (E4) {$q^\ominus$};
					\node[violet,node,right=6mm of E4] (N5) {};
					\node[violet,hyperedge,right=4mm of N5] (E5) {$r^\ominus$};
					\node[hyperedge,below right =-1mm and -27mm of O] (E2) {$\tau(p\cdot r)^\oplus$};
					\node[violet,node,below=8mm of O] (N3) {};
					\node[violet,hyperedge,below=3mm of N3] (E3) {$\beta$};
					\draw[-] (N1) to[bend left=20] node[above] {\scriptsize 1} (E1);
					\draw[-] (E1) -- node[above] {\scriptsize 2} (N4);
					\draw[violet,-] (N4) -- node[above] {\scriptsize 1} (E4);
					\draw[violet,-] (E4) -- node[above] {\scriptsize 2} (N5);
					\draw[violet,-] (N5) -- node[above] {\scriptsize 1} (E5);
					\draw[-] (N4) -- node[above] {\scriptsize 2} (E1);
					\draw[violet,-] (E5) to[bend left=20] node[above right] {\scriptsize 2} (N2);
					\draw[-] (N2) to[bend left=10] node[above] {\scriptsize 2} (E2);
					\draw[-] (E2) to[bend left=10] node[above] {\scriptsize 1} (N1);
					\draw[violet,-] (E3) -- (N3);
					\draw[-] (E1) to[bend right=20] node[below] {\scriptsize 3} (N3);
					\draw[-] (E2) to[bend right=22] node[below left] {\scriptsize 3} (N3);
					\draw[violet,-] (E4) to[bend left=10] node[above left] {\scriptsize 3} (N3);
					\draw[violet,-] (E5) to[bend left=35] node[right] {\scriptsize 3} (N3);
					\draw[->] (N1) to[out=-90,in=210] node[below] {$\kappa$} (N3);
					\draw[violet,->] (N2) to[out=-90,in=-30] node[below] {$\kappa$} (N3);
					\draw[violet,->] (N5) to[bend left=25] node[right] {$\kappa$} (N3) ;
					\draw[violet,->] (N4) to[bend right=8] node[left] {$\kappa$} (N3);
	}}}}
	\;\;\Rightarrow\;\;
	$$
	$$
	\Rightarrow\;\;
	\vcenter{\hbox{{\tikz[baseline=.1ex]{
					\node[node,label=above:{\scriptsize $(1)$}] (N1) {};
					\node[node, below=6mm of N1,label=below:{\scriptsize $(2)$}] (N2) {};
					\node[purple,hyperedge, below=5.5mm of N2] (E) {$q_2$};
					\draw[->] (N1) -- node[right] {$\varphi$} (N2);
	}}}}
	\;
	\vcenter{\hbox{{\tikz[baseline=.1ex]{
					\node[] (O) {};
					\node[purple,node,left=19mm of O] (N1) {};
					\node[purple,node,right=19mm of O] (N2) {};
					\node[purple,node,above=3mm of O] (N4) {};
					\node[hyperedge,left= 4mm of N4] (E1) {$\tau(p/q)^\ominus$};
					\node[purple,hyperedge,right=4mm of N4] (E4) {$\tau(q\cdot r)^\ominus$};
					\node[purple,hyperedge,below right =-1mm and -5mm of O] (E2) {$\tau(p\cdot r)^\oplus$};
					\node[purple,node,below=8mm of O] (N3) {};
					\node[purple,hyperedge,below=3mm of N3] (E3) {$\beta$};
					\draw[-] (N1) to[bend left=10] node[above] {\scriptsize 1} (E1);
					\draw[-] (E1) -- node[above] {\scriptsize 2} (N4);
					\draw[purple,-] (N4) -- node[above] {\scriptsize 1} (E4);
					\draw[purple,-] (E4) to[bend left=10] node[above right] {\scriptsize 2} (N2);
					\draw[purple,-] (N2) to[bend left=10] node[below] {\scriptsize 2} (E2);
					\draw[purple,-] (E2) to[bend left=10] node[below left] {\scriptsize 1} (N1);
					\draw[purple,-] (E3) -- (N3);
					\draw[-] (E1) to[bend right=30] node[below] {\scriptsize 3} (N3);
					\draw[purple,-] (E2) -- node[right] {\scriptsize 3} (N3);
					\draw[purple,->] (N1) to[out=-90,in=180] node[below] {$\kappa$} (N3);
					\draw[purple,->] (N2) to[out=-90,in=0] node[below] {$\kappa$} (N3);
					\draw[purple,->] (N4) 
					to[out=-150,in=80] (-5mm,-0.5mm) node[left] {$\kappa$}
					to[out=-100,in=135] (N3)
					;
	}}}}
	\Rightarrow
	\vcenter{\hbox{{\tikz[baseline=.1ex]{
					\node[node,label=above:{\scriptsize $(1)$}] (N1) {};
					\node[node, below=6mm of N1,label=below:{\scriptsize $(2)$}] (N2) {};
					\node[hyperedge, below=5.5mm of N2] (E) {$q_2$};
					\draw[->] (N1) -- node[right] {$\varphi$} (N2);
	}}}}
	\;
	\vcenter{\hbox{{\tikz[baseline=.1ex]{
					\node[] (O) {};
					\node[node,left=20mm of O] (N1) {};
					\node[node,right=15mm of O] (N2) {};
					\node[hyperedge,above right=1mm and 2.5mm of N1] (E1) {$\tau(p/q)^\ominus$};
					\node[hyperedge,below left=-1mm and 3.5mm of N2] (E2) {$\tau((p\cdot r)/(q\cdot r))^\oplus$};
					\node[node,below=6mm of E2] (N3) {};
					\node[hyperedge,below=3mm of N3] (E3) {$\beta$};
					\draw[-] (N1) to[bend left=10] node[above] {\scriptsize 1} (E1);
					\draw[-] (E1) to[bend left=10] node[above right] {\scriptsize 2} (N2);
					\draw[-] (N2) to[bend left=10] node[below] {\scriptsize 2} (E2);
					\draw[-] (E2) to[bend left=10] node[below left] {\scriptsize 1} (N1);
					\draw[-] (E3) -- (N3);
					\draw[-] (E1) to[out=-120,in=180] node[right] {\scriptsize 3} (N3);
					\draw[-] (E2) -- node[right] {\scriptsize 3} (N3);
					\draw[->] (N1) to[out=-90,in=180] node[below] {$\kappa$} (N3);
					\draw[->] (N2) to[out=-90,in=0] node[below] {$\kappa$} (N3);
	}}}}
	$$
	
	The rules used in this derivation are as follows:
	
	$$
	r^{\tau(p\cdot r)}_R = \left(
	{\color{blue}
	\vcenter{\hbox{{\tikz[baseline=.1ex]{
					\node[hyperedge] (E1) {$p^\oplus$};
					\node[node,left=6mm of E1, label=above:{\scriptsize $(1)$}] (N2) {};
					\node[node,right=6mm of E1, label=above:{\scriptsize $(2)$}] (N3) {};
					\node[node,below=6mm of E1, label=right:{\scriptsize $(3)$}] (N4) {};
					\node[hyperedge,below =3mm of N4] (E2) {$\beta$};
					\draw[-] (E1) -- node[above] {\scriptsize 1} (N2);
					\draw[-] (E1) -- node[above] {\scriptsize 2} (N3);
					\draw[-] (E1) -- node[right] {\scriptsize 3} (N4);
					\draw[-] (N4) -- (E2);
					\draw[->] (N2) -- node[left] {$\kappa$} (N4);
					\draw[->] (N3) -- node[right] {$\kappa$} (N4);
	}}}}
	\vcenter{\hbox{{\tikz[baseline=.1ex]{
					\node[hyperedge] (E1) {$r^\oplus$};
					\node[node,left=6mm of E1, label=above:{\scriptsize $(4)$}] (N2) {};
					\node[node,right=6mm of E1, label=above:{\scriptsize $(5)$}] (N3) {};
					\node[node,below=6mm of E1, label=right:{\scriptsize $(6)$}] (N4) {};
					\node[hyperedge,below =3mm of N4] (E2) {$\beta$};
					\draw[-] (E1) -- node[above] {\scriptsize 1} (N2);
					\draw[-] (E1) -- node[above] {\scriptsize 2} (N3);
					\draw[-] (E1) -- node[right] {\scriptsize 3} (N4);
					\draw[-] (N4) -- (E2);
					\draw[->] (N2) -- node[left] {$\kappa$} (N4);
					\draw[->] (N3) -- node[right] {$\kappa$} (N4);
	}}}}
	\vcenter{\hbox{{\tikz[baseline=.1ex]{\node[hyperedge] {$q_2$};}}}}
	}
	\;
	\middle|
	\vcenter{\hbox{{\tikz[baseline=.1ex]{
					\node[hyperedge] (E1) {$\tau(p\cdot r)^\oplus$};
					\node[node,above=3.5mm of E1, label=above:{\scriptsize $(2)\,(4)$}] (N1) {};
					\node[node,left=6mm of E1, label=above:{\scriptsize $(1)$}] (N2) {};
					\node[node,right=6mm of E1, label=above:{\scriptsize $(5)$}] (N3) {};
					\node[node,below=5mm of E1, label=below right:{\scriptsize $(6)$}, label=below left:{\scriptsize $(3)$}] (N4) {};
					\node[hyperedge,below =3.5mm of N4] (E2) {$\beta$};
					\draw[-] (E1) -- node[above] {\scriptsize 1} (N2);
					\draw[-] (E1) -- node[above] {\scriptsize 2} (N3);
					\draw[-] (E1) -- node[right] {\scriptsize 3} (N4);
					\draw[-] (N4) -- (E2);
					\draw[->] (N2) -- node[left] {$\kappa$} (N4);
					\draw[->] (N3) -- node[right] {$\kappa$} (N4);
					\draw[->] (N1) 
					to[out=180,in=90] (-17mm,0mm) node[left] {$\kappa$}
					to[out=270,in=180] (N4)
					;
	}}}}
	\vcenter{\hbox{{\tikz[baseline=.1ex]{\node[hyperedge] {$q_2$};}}}}
	\right)
	$$
	
	$$
	r^{\tau(p/q)}_L = 
	\left(
	{\color{red}
		\vcenter{\hbox{{\tikz[baseline=.1ex]{
					\node[hyperedge] (E1) {$p^\ominus$};
					\node[node,left=6mm of E1, label=above:{\scriptsize $(1)$}] (N2) {};
					\node[node,right=6mm of E1, label=above:{\scriptsize $(2)$}] (N3) {};
					\node[node,below=6mm of E1, label=right:{\scriptsize $(3)$}] (N4) {};
					\node[hyperedge,below =3mm of N4] (E2) {$\beta$};
					\draw[-] (E1) -- node[above] {\scriptsize 1} (N2);
					\draw[-] (E1) -- node[above] {\scriptsize 2} (N3);
					\draw[-] (E1) -- node[right] {\scriptsize 3} (N4);
					\draw[-] (N4) -- (E2);
					\draw[->] (N2) -- node[left] {$\kappa$} (N4);
					\draw[->] (N3) -- node[right] {$\kappa$} (N4);
	}}}}
	\vcenter{\hbox{{\tikz[baseline=.1ex]{
					\node[hyperedge] (E1) {$q^\oplus$};
					\node[node,left=6mm of E1, label=above:{\scriptsize $(4)$}] (N2) {};
					\node[node,right=6mm of E1, label=above:{\scriptsize $(5)$}] (N3) {};
					\node[node,below=6mm of E1, label=right:{\scriptsize $(6)$}] (N4) {};
					\node[hyperedge,below =3mm of N4] (E2) {$\beta$};
					\draw[-] (E1) -- node[above] {\scriptsize 1} (N2);
					\draw[-] (E1) -- node[above] {\scriptsize 2} (N3);
					\draw[-] (E1) -- node[right] {\scriptsize 3} (N4);
					\draw[-] (N4) -- (E2);
					\draw[->] (N2) -- node[left] {$\kappa$} (N4);
					\draw[->] (N3) -- node[right] {$\kappa$} (N4);
	}}}}
	\vcenter{\hbox{{\tikz[baseline=.1ex]{\node[hyperedge] {$q_2$};}}}}
	}
	\;
	\middle|
	\vcenter{\hbox{{\tikz[baseline=.1ex]{
					\node[hyperedge] (E1) {$\tau(p/q)^\ominus$};
					\node[node,right=14mm of E1, label=above:{\scriptsize $(2)\,(5)$}] (N1) {};
					\node[node,left=6mm of E1, label=above:{\scriptsize $(1)$}] (N2) {};
					\node[node,right=6mm of E1, label=above:{\scriptsize $(4)$}] (N3) {};
					\node[node,below=6mm of E1, label=below right:{\scriptsize $(6)$}, label=below left:{\scriptsize $(3)$}] (N4) {};
					\node[hyperedge,below =4mm of N4] (E2) {$\beta$};
					\draw[-] (E1) -- node[above] {\scriptsize 1} (N2);
					\draw[-] (E1) -- node[above] {\scriptsize 2} (N3);
					\draw[-] (E1) -- node[right] {\scriptsize 3} (N4);
					\draw[-] (N4) -- (E2);
					\draw[->] (N2) -- node[left] {$\kappa$} (N4);
					\draw[->] (N3) -- node[right] {$\kappa$} (N4);
					\draw[->] (N1) 
					to[bend left=25] node[below] {$\kappa$} (N4)
					;
	}}}}
	\!\!\!\!
	\vcenter{\hbox{{\tikz[baseline=.1ex]{\node[hyperedge] {$q_2$};}}}}
	\right)
	$$
	
	$$
	r^{\tau(q\cdot r)}_L = 
	\left(
	{\color{violet}
		\vcenter{\hbox{{\tikz[baseline=.1ex]{
					\node[node] (N0) {};
					\node[hyperedge,left=5mm of N0] (E1) {$q^\ominus$};
					\node[node,left=5mm of E1, label=above:{\scriptsize $(1)$}] (N1) {};
					\node[hyperedge,right=5mm of N0] (E2) {$r^\ominus$};
					\node[node,right=5mm of E2, label=above:{\scriptsize $(2)$}] (N2) {};
					\node[node,below=9mm of N0, label=below right:{\scriptsize $(3)$}] (N3) {};
					\node[hyperedge,below =4mm of N3] (E3) {$\beta$};
					\draw[-] (N1) -- node[above] {\scriptsize 1} (E1);
					\draw[-] (E1) -- node[above] {\scriptsize 2} (N0);
					\draw[-] (N0) -- node[above] {\scriptsize 1} (E2);
					\draw[-] (E2) -- node[above] {\scriptsize 2} (N2);
					\draw[-] (E1) -- node[left] {\scriptsize 3} (N3);
					\draw[-] (E2) -- node[right] {\scriptsize 3} (N3);
					\draw[-] (E3) -- (N3);
					\draw[->] (N0) -- node[above left] {$\kappa$} (N3);
					\draw[->] (N1) to[bend right=20] node[left] {$\kappa$} (N3);
					\draw[->] (N2) to[bend left=20] node[right] {$\kappa$} (N3);
	}}}}
	\vcenter{\hbox{{\tikz[baseline=.1ex]{\node[hyperedge] {$q_2$};}}}}
	}
	\;
	\middle|
	\vcenter{\hbox{{\tikz[baseline=.1ex]{
					\node[hyperedge] (E1) {$\tau(q\cdot r)^\ominus$};
					\node[node,left=6mm of E1, label=above:{\scriptsize $(1)$}] (N2) {};
					\node[node,right=6mm of E1, label=above:{\scriptsize $(2)$}] (N3) {};
					\node[node,below=6mm of E1, label=below right:{\scriptsize $(3)$}] (N4) {};
					\node[hyperedge,below =4mm of N4] (E2) {$\beta$};
					\draw[-] (E1) -- node[above] {\scriptsize 1} (N2);
					\draw[-] (E1) -- node[above] {\scriptsize 2} (N3);
					\draw[-] (E1) -- node[right] {\scriptsize 3} (N4);
					\draw[-] (N4) -- (E2);
					\draw[->] (N2) to[bend right=20] node[left] {$\kappa$} (N4);
					\draw[->] (N3) to[bend left=20] node[right] {$\kappa$} (N4);
	}}}}
	\vcenter{\hbox{{\tikz[baseline=.1ex]{\node[hyperedge] {$q_2$};}}}}
	\right)
	$$
	\begin{multline*}
	r^{\tau(p\cdot r)/(q\cdot r)}_R = 
	\\
	\left(
	{\color{purple}
		\vcenter{\hbox{{\tikz[baseline=.1ex]{
					\node[node] (N0) {};
					\node[hyperedge,left=4mm of N0] (E1) {$\tau(p\cdot r)^\oplus$};
					\node[node,left=4mm of E1, label=above:{\scriptsize $(1)$}] (N1) {};
					\node[hyperedge,right=4mm of N0] (E2) {$\tau(q\cdot r)^\ominus$};
					\node[node,right=4mm of E2, label=above:{\scriptsize $(2)$}] (N2) {};
					\node[node,below=7mm of N0, label=below right:{\scriptsize $(3)$}] (N3) {};
					\node[hyperedge,below =4mm of N3] (E3) {$\beta$};
					\draw[-] (N1) -- node[above] {\scriptsize 1} (E1);
					\draw[-] (E1) -- node[above] {\scriptsize 2} (N0);
					\draw[-] (N0) -- node[above] {\scriptsize 2} (E2);
					\draw[-] (E2) -- node[above] {\scriptsize 1} (N2);
					\draw[-] (E1) -- node[left] {\scriptsize 3} (N3);
					\draw[-] (E2) -- node[right] {\scriptsize 3} (N3);
					\draw[-] (E3) -- (N3);
					\draw[->] (N0) -- node[above left] {$\kappa$} (N3);
					\draw[->] (N1) to[bend right=20] node[below] {$\kappa$} (N3);
					\draw[->] (N2) to[bend left=20] node[below] {$\kappa$} (N3);
	}}}}
	\!\!
	\vcenter{\hbox{{\tikz[baseline=.1ex]{\node[hyperedge] {$q_2$};}}}}
	}
	\;
	\middle|
	\vcenter{\hbox{{\tikz[baseline=.1ex]{
					\node[hyperedge] (E1) {$\tau((p \cdot r)/(q\cdot r))^\ominus$};
					\node[node,left=6mm of E1, label=above:{\scriptsize $(1)$}] (N2) {};
					\node[node,right=6mm of E1, label=above:{\scriptsize $(2)$}] (N3) {};
					\node[node,below=4mm of E1, label=below left:{\scriptsize $(3)$}] (N4) {};
					\node[hyperedge,below =4mm of N4] (E2) {$\beta$};
					\draw[-] (E1) -- node[above] {\scriptsize 1} (N2);
					\draw[-] (E1) -- node[above] {\scriptsize 2} (N3);
					\draw[-] (E1) -- node[right] {\scriptsize 3} (N4);
					\draw[-] (N4) -- (E2);
					\draw[->] (N2) to[bend right=20] node[below] {$\kappa$} (N4);
					\draw[->] (N3) to[bend left=20] node[below] {$\kappa$} (N4);
	}}}}
	\!\!
	\vcenter{\hbox{{\tikz[baseline=.1ex]{\node[hyperedge] {$q_2$};}}}}
	\right)
	\end{multline*}
\end{example}

\end{document}